\documentclass[conference]{IEEEtran}
\IEEEoverridecommandlockouts

\usepackage{cite}
\usepackage{amsmath,amssymb,amsfonts}
\usepackage{graphicx}
\usepackage{textcomp}
\usepackage{xcolor}
\usepackage{hyperref}

\usepackage{xspace}
\usepackage{amsthm}
\usepackage{complexity}
\usepackage{booktabs}
\usepackage{algorithm}
\usepackage[noend]{algpseudocode}
\usepackage[capitalize]{cleveref}
\usepackage{thm-restate}
\usepackage[normalem]{ulem}

\usepackage{tikz}
\usetikzlibrary{decorations.pathreplacing}

\newcommand{\sentence}{\Gamma}

\newcommand{\formula}{\alpha}

\newcommand{\vocabulary}{\mathcal{P}}

\newcommand{\fotwo}{\ensuremath{\mathbf{FO}^2}\xspace}
\newcommand{\ctwo}{\ensuremath{\mathbf{C}^2}\xspace}

\newcommand{\domain}{\Delta}

\newcommand{\vecn}{\mathbf{n}}

\newcommand{\fomodel}{\ensuremath{\mathcal{M}}\xspace}

\newcommand{\structure}{\mathcal{A}}

\newcommand{\usub}{\ensuremath{\structure_u}}
\newcommand{\bsub}{\ensuremath{\structure_b}}
\newcommand{\logO}{\ensuremath{\widetilde{O}}}

\newcommand{\tsentence}{\ensuremath{\widehat{\sentence}}}
\newcommand{\tvecn}{\ensuremath{\widehat{\vecn}}}
\newcommand{\targete}{\ensuremath{e^*}}
\newcommand{\substr}{\ensuremath{\mathcal{S}}}

\newcommand{\confsat}{\ensuremath{\mathsf{SAT_{\! cfg}}}}
\newcommand{\subsat}{\ensuremath{\mathsf{SAT_{\! str}}}}

\renewcommand{\theenumi}{(\arabic{enumi}}

\newtheorem{definition}{Definition}
\newtheorem{example}{Example}
\newtheorem{remark}{Remark}
\newtheorem{problem}{Problem}
\newtheorem{theorem}{Theorem}
\newtheorem{lemma}{Lemma}
\newtheorem{proposition}{Proposition}
\newtheorem{corollary}{Corollary}

\def\BibTeX{{\rm B\kern-.05em{\sc i\kern-.025em b}\kern-.08em
    T\kern-.1667em\lower.7ex\hbox{E}\kern-.125emX}}
\begin{document}

\title{Model Enumeration of Two-Variable Logic with Quadratic Delay Complexity
\thanks{
\textsuperscript{*}Corresponding author.

Qiaolan Meng and Juhua Pu are supported by the National Natural Science Foundation of China (62177002). OK's work was supported by the Czech Science Foundation project ``The Automatic Combinatorialist'' (24-11820S).
Qiaolan Meng and Yuyi Wang are supported by the Natural Science Foundation of Hunan Province, China (Grant No.\ 2024JJ5128).} 
}

\author{
\IEEEauthorblockN{Qiaolan Meng}
\IEEEauthorblockA{
\textit{State Key Laboratory of } \\
\textit{Software Development Environment,}\\
\textit{Beihang University}\\
\textit{\& National Research Center}\\
\textit{for Educational Materials}\\
Beijing, China \\
mengql1606@buaa.edu.cn}
\and
\IEEEauthorblockN{Juhua Pu}
\IEEEauthorblockA{
\textit{State Key Laboratory of } \\
\textit{Software Development Environment,}\\
\textit{Beihang University}\\
\textit{\& National Research Center}\\
\textit{for Educational Materials}\\
Beijing, China \\
pujh@buaa.edu.cn}
\and
\IEEEauthorblockN{Hongting Niu}
\IEEEauthorblockA{
\textit{Hangzhou International Innovation Institute}\\
\textit{\& State Key Laboratory of Software} \\
\textit{Development Environment, Beihang University}\\
Hangzhou, China \\
niuhongting@buaa.edu.cn}
\and
\IEEEauthorblockN{Yuyi Wang}
\IEEEauthorblockA{\textit{CRRC Zhuzhou Insitute}\\
\textit{\&  Tengen Intelligence Institute} \\
Zhuzhou, China \\
yuyiwang920@gmail.com}
\and
\IEEEauthorblockN{Yuanhong Wang$^*$}
\IEEEauthorblockA{\textit{School of Artificial Intelligence,} \\
\textit{Jilin University}\\
Jilin, China \\
lucienwang@jlu.edu.cn}
\and
\IEEEauthorblockN{Ondřej Kuželka}
\IEEEauthorblockA{\textit{Faculty of Electrical Engineering, } \\
\textit{Czech Technical University in Prague}\\
Prague, Czech Republic \\
ondrej.kuzelka@fel.cvut.cz}
}
\maketitle
\begin{abstract}
We study the model enumeration problem of the function-free, finite domain fragment of first-order logic with two variables (\fotwo{}). 
Specifically, given an \fotwo{} sentence $\sentence$ and a positive integer $n$, how can one enumerate all the models of $\sentence$ over a domain of size $n$?
In this paper, we devise a novel algorithm to address this problem.
The delay complexity, the time required between producing two consecutive models, of our algorithm is quadratic in the given domain size $n$ (up to logarithmic factors) when the sentence is fixed.
This complexity is almost optimal since the interpretation of binary predicates in any model requires at least $\Omega(n^2)$ bits to represent. 

\end{abstract}

\begin{IEEEkeywords}
model enumeration, first-order logic, delay complexity, satisfiability problem.
\end{IEEEkeywords}

\section{Introduction}

The enumeration problem plays a critical role in various applications, such as database query evaluation~\cite{DBLP:journals/tocl/DurandG07,durand2014enumerating}, constraint satisfaction~\cite{crawford2011hyperheuristic,schnoor2007enumerating}, and exhaustive verification~\cite{nievergelt2000exhaustive,stern1997algorithmic}.
In this paper, we study the model enumeration problem for first-order logic: Given a first-order logic sentence $\sentence$ and a positive integer $n$, the goal is to enumerate all the models of $\sentence$ over a domain of size $n$.

Model enumeration falls within the domain of finite model theory and is closely related to two other important problems in the field: \emph{model counting}~\cite{beame_symmetric_2015-1} and \emph{model sampling}~\cite{wang_exact_2023-1}. 
Given the same input, i.e., a first-order logic sentence $\sentence$ and a positive integer $n$, the model counting problem asks for the total number of distinct models of $\sentence$ over a domain of size $n$, while the model sampling problem aims to generate models uniformly at random.
For these two problems, research primarily focuses on \emph{data complexity}~\cite{beame_symmetric_2015-1}, which refers to the time required to solve the problem as a function of the input integer $n$, with the sentence $\sentence$ fixed.
Various positive and negative results w.r.t.\ the data complexity have been established for model counting and model sampling across different fragments of first-order logic.
However, similar complexity results are not available for model enumeration, which remains relatively underexplored. 

We consider the \emph{delay complexity} of model enumeration, which is the time required between producing two consecutive models.
In this paper, the delay complexity is measured in terms of the domain size $n$ with the sentence fixed, analogous to the data complexity of model counting and sampling.
This type of delay complexity also relates to the enumeration problems of combinatorial structures~\cite{DBLP:phd/ethos/Goldberg91,DBLP:journals/jmlr/RamonN09}, wherein the complexity is w.r.t.\ the size of the output while the constraints on combinatorial structures, e.g., connectivity, are fixed.


In this paper, we are particularly interested in the model enumeration problem for the fragment of first-order logic restricted to two variables (\fotwo{}).
The key motivation is that \fotwo{} is a well-studied fragment of first-order logic due to its balance of expressiveness and tractability.
The data complexity of both model counting and model sampling for \fotwo{} is polynomial in the domain size $n$ \cite{beame_symmetric_2015-1,wang_exact_2023-1}, so one might naturally expect that model enumeration for \fotwo{} could also be solvable in polynomial time with respect to $n$. 
In fact, we demonstrate that this result can be even stronger, i.e., the delay complexity can be quadratic in the domain size $n$ ignoring the logarithmic factors.

Although the data complexity of model counting and model sampling problems for \fotwo{} have been proved to be polynomial in the domain size, their complexity includes an exponential factor that depends on the number $s$ of predicates in the sentence. 
Consequently, directly applying existing techniques from counting and sampling algorithms for \fotwo{} to model enumeration would not yield an enumeration algorithm with quadratic delay.
In particular, the existing \fotwo{} sampling algorithms~\cite{wang_exact_2023-1, wang2024lifted}, the most related to our enumeration problem, use a technique called \emph{domain recursion} that samples the related atomic facts of each domain element in a recursive manner.
If we blindly applied the domain recursion technique to model enumeration, the delay complexity would not be better than $O(n^s)$, which is evidently beyond the quadratic bound for non-trivial sentences.

\subsection{Our Contributions}
  In this paper, we present a novel enumeration algorithm for \fotwo{}.
The algorithm takes inspiration from the sampling algorithm for \fotwo{}~\cite{wang_exact_2023-1} while avoiding the exponential factor.
The algorithm uses a refined domain recursion scheme, where each recursion step enumerates the atomic facts of a pair of domain elements. Each step of this scheme requires only a constant number of calls to a satisfiability oracle, whereas the sampling algorithm needs multiple calls to a computationally much more expensive model counting oracle. 
We prove that our algorithm has a \emph{quadratic delay complexity} in the domain size $n$ (ignoring the logarithmic factors introduced by arithmetic operations), where the delay contains the time to construct and output a model. This result is formally stated in the following theorem.
\begin{restatable}{theorem}{quadraticdelay}
  \label{th:quadratic_delay}
  Let $\sentence$ be an \fotwo{} sentence and $n$ be a positive integer. 
  The delay complexity of model enumeration of $\sentence$ over a domain of size $n$ is $\logO(n^2)$.
\end{restatable}
The complexity that we obtain is \emph{almost optimal} since the interpretation of binary predicates over a domain of size $n$ inherently requires at least $\Omega(n^2)$ bits to represent.
We further extend our enumeration algorithm to handle the \fotwo{} sentence with \emph{the equality} relation, and show that the delay complexity remains quadratic in the presence of the equality.
\begin{theorem}
  \label{th:quadratic_delay_equality}
  Let $\sentence$ be an \fotwo{} sentence with the equality predicate and $n$ be a positive integer. 
  The delay complexity of model enumeration of $\sentence$ over a domain of size $n$ is $\logO(n^2)$.
\end{theorem}

In addition, we implement our enumeration algorithm and conduct experiments on several benchmark \fotwo{} sentences. 
The result demonstrates that our algorithm achieves consistently low delay and stable performance, outperforming the All-SAT-based enumeration method.
The code as well as the experimental results are available at \url{https://github.com/MengQiaolan/model\_enumeration\_fo2}.

\subsection{Related Work}

There has been extensive research on the finite model theory.
One of the core problems is the satisfiability problem~\cite{Scott1962}, which seeks to determine whether a given formula has a finite model.
The satisfiability problem of \fotwo{}, denoted as $\SAT(\fotwo)$, is known to be decidable~\cite{mortimer1975languages,1997decision}, and its exact complexity is NEXPTIME-complete~\cite{Scott1962}.
Several approaches have been proposed to solve $\SAT(\fotwo)$ efficiently~\cite{de2001resolution,tonytang2021towards}.
Beyond satisfiability, model counting and model sampling are two closely related problems to model enumeration that have also been studied in the literature. 
While model counting and sampling in general first-order logic are
intractable (under plausible complexity-theoretic assumptions), there exist algorithms for \fotwo{} that solve both problems in polynomial time with respect to the domain size~\cite{van2021faster, wang2024lifted}.
In particular, the scope of counting tractability has been extended to the fragments beyond \fotwo{}, such as \ctwo~\cite{kuzelka2021weighted}, $\mathbf{S^2FO^2}$~\cite{kazemi2017domain}, $\mathbf{S^2RU}$~\cite{kazemi2017domain}, etc.
There is also literature on extending the sampling tractability to \ctwo{}~\cite{wang2024lifted} and tree axiom~\cite{wang_domain-lifted_2019}.

Enumeration has been widely studied in several fields. 
A relevant prior work to ours is the model enumeration for propositional logic, also called All-SAT (All Satisfying Assignments).
Classical All-SAT solvers~\cite{yu2014all,toda2016implementing} achieve enumeration by iteratively calling SAT solvers to find a model and adding blocking clauses to the given formula to rule out the found models.
This approach is typically inefficient since an exponential number of blocking clauses may be added to the formula.
Then, more efficient algorithms have been developed, such as those based on the DPLL~\cite{nieuwenhuis2006solving} or CDCL~\cite{heule2011cube} solvers, which avoid explicit blocking clauses by some sophisticated techniques.

The quadratic delay complexity obtained in this paper makes the enumeration problem of \fotwo{} fall within the class of \emph{fixed-parameter tractable problems} (\FPT{}) when the sentence size is treated as a parameter to the problem.

\section{Preliminaries}

We first introduce some basic concepts and notations used in this paper.

The partitions used in this paper are all ordered partitions, i.e., the order of the parts matters.
We use $[n]$ to denote the set $\{1,2,\dots,n\}$ for any positive integer $n$.
We also use the bold symbol $\vecn$ to denote a vector $(n_1,\dots,n_k)$, and $|\vecn|$ to denote the summation of the elements in $\vecn$, that is, $|\vecn| = \sum_{i=1}^k n_i$.

\subsection{First-Order Logic}

This paper focuses on the \emph{function-free} and \emph{finite-domain} fragment of first-order logic.
An \emph{atom} of arity $k$ is expressed as $P(x_1,\dots,x_k)$, where $P/k$ is a predicate from a given predicate vocabulary, and $x_1,\dots,x_k$ are logical variables or constants.
A \emph{literal} is either an atom or its negation.
A \emph{formula} is inductively defined as an atom, the negation of another formula, or the conjunction or disjunction of two formulas.
Additionally, a formula may be enclosed by one or more quantifiers, such as $\forall x$ or $\exists x$, and still be considered a formula.
A logical variable in a formula is referred to as \emph{free} if it is not bound by any quantifier.
A formula that contains no free variables is termed a \emph{sentence}.
The predicate vocabulary of a formula $\formula$ is denoted as $\mathcal{P}_\formula$.
A formula is said to be \emph{ground} if it contains no variables.
A \emph{structure} $\structure$ of a predicate vocabulary $\mathcal{P}$ is a tuple $(\domain, f)$, where $\domain$ is a \emph{finite} domain and $f$ is a mapping from predicates $P/k$ to $\domain^k$.
In this paper, a structure is often written as a set of ground literals interpreted by $f$: $\structure = \bigcup_{P/k \in \mathcal{P}_\formula} \{P(a_1,\dots,a_k) \mid (a_1,\dots,a_k) \in f(P)\} \cup \{\neg P(a_1,\dots,a_k) \mid (a_1,\dots,a_k) \notin f(P)\}$.
The satisfaction relation $\models$ is defined by the standard semantics of first-order logic.
A \emph{model} \fomodel is a structure that satisfies a given formula.

\subsection{Types, Substructures and Configuration}
\label{sec:types_substructures_config}

This paper focuses on the fragment of first-order logic restricted to two logical variables, denoted as \fotwo{}.
Any \fotwo{} sentence (without the equality predicate) can be transformed into the \emph{Scott normal form (SNF)}~\cite{Scott1962}
\begin{equation}
  \sentence = \forall x \forall y: \phi(x,y) \land \bigwedge_{k\in[m]} \forall x\exists y: \beta_k(x,y),
  \label{eq:snf}
\end{equation}
where $\phi(x,y)$ are quantifier-free formulas and $\beta_k$ are binary predicates.
The transformation preserves the models, that is, there is a bijective mapping between the models of the original sentence and those of its SNF sentence~\cite[Appendix A.1]{wang2024lifted}.
Specifically, for each model of the SNF sentence, we can obtain the corresponding model of the original sentence by removing and replacing the literals of the auxiliary predicates, and the complexity of this process is linear in the size of the sentence~\cite{1997decision}.
Thus enumerating models of an \fotwo{} sentence is equivalent to enumerating models of its SNF sentence.
In the rest of this paper, we assume that the \fotwo{} sentences have been transformed into the SNF.

Next, we need two important concepts: \emph{1-types} and \emph{2-types}.

\begin{definition}[Types]
  \label{def:types}
  Given a predicate vocabulary $\vocabulary$, a \emph{1-type} $\tau$ of $\vocabulary$ is a maximally consistent set of literals formed by $\vocabulary$ involving a \emph{single} variable $x$.\footnote{Here, a set of literals is \emph{maximally consistent} if it does not contain any contradictory literals and cannot be extended by adding any other literal.}
  A \emph{2-type} $\pi$ is a maximally consistent set of literals formed by $\vocabulary$ each involving \emph{two} variables $x$ and $y$.\footnote{Note that while the definition of the 1-type here is consistent with standard definitions in logic literature, the definition of the 2-type diverges from the usual one. Specifically, our 2-type does not include information about the 1-types of its two constituent variables.}
  A \emph{type} is either a 1-type or a 2-type.
\end{definition}
The 1-types and 2-types of an \fotwo sentence $\sentence$ are defined as that of its predicate vocabulary $\mathcal{P}_\sentence$, which we denote by $U_\sentence$ and $B_\sentence$, respectively.
When the context is clear, the subscript $\sentence$ is omitted.
For convenience, a type is often identified with the conjunction of its constituent literals, expressed as a formula $\tau(x)$ or $\pi(x,y)$.
Its ground $\tau(e)$ or $\pi(e,e')$ is also viewed interchangeably as a set of ground literals for some constants $e$ and $e'$.
A 1-type $\tau$ is \emph{compatible with} $\sentence$ if $\tau(e) \models \phi(e,e)$ for some constant $e$.
A 2-type $\pi$ is \emph{compatible with} $\sentence$ and 1-type pair $\{\tau_i, \tau_j\}$ if $\tau_i(e) \land \tau_j(e') \land \pi(e,e') \models \phi(e,e')$ for some constants $e$ and $e'$.
Similarly, we omit the notation $\sentence$ and simply say that a 1-type is compatible, and a 2-type is compatible with a pair of 1-types, when the context is clear.

In a structure $\structure$, an element $e \in \domain$ is said to \emph{realize} the 1-type $\tau$ if $\structure \models \tau(e)$.
Obviously, every element realizes exactly one 1-type.
We say that a pair of elements $\{e, e'\}$ realizes the 2-type $\pi$ if $\structure \models \pi(e, e')$.
Similarly, every pair of elements realizes exactly one 2-type.

\begin{example}
  \label{ex:types}
  Consider the \fotwo{} sentence that describes 2-colored graphs:
  \begin{align*}
    \sentence_{col} = \ &\forall x: (\neg R(x) \lor \neg B(x)) \land (R(x) \lor B(x)) \land \\
    &\forall x \forall y: E(x, y) \rightarrow (R(x) \land B(y)) \lor (B(x) \land R(y)) \land \\
    &\forall x \forall y: E(x, y) \rightarrow E(y, x),
  \end{align*}
  where the unary predicates $R$ and $B$ represent the red and blue colors, respectively.
  There are $2^3$ 1-types, but only two of them are compatible with $\sentence_{col}$ : $\tau_1(x) = R(x) \land \neg B(x) \land \neg E(x,x)$ and $\tau_2(x) = \neg R(x) \land B(x) \land \neg E(x,x)$.
  Similarly, there are two compatible 2-types: $\pi_1(x, y) = E(x,y) \land E(y,x)$ and $\pi_2(x, y) = \neg E(x,y) \land \neg E(y,x)$.
  For the structure (actually a model) over a domain $\domain = \{e_1, e_2\}$:
  \begin{gather*}
    \{ R(e_1), \neg B(e_1), \neg E(e_1,e_1), \neg R(e_2), B(e_2), \neg E(e_2,e_2), \\
    \neg E(e_1,e_2), \neg E(e_2,e_1) \},
  \end{gather*}
  the elements $e_1$ and $e_2$ realize the 1-types $\tau_1$ and $\tau_2$, respectively, and the pair $\{e_1, e_2\}$ realizes the 2-type $\pi_2$. 
\end{example}

Let $\mathcal{A}$ be a structure.
We call a subset $S$ of $\mathcal{A}$ a \emph{substructure} of $\mathcal{A}$.
A substructure $\substr$ is said to be $\domain$-\emph{consistent} w.r.t. $\sentence$ if $\substr$ is a substructure of a model of $\sentence$ over $\domain$.
Based on the types, we can define two special substructures: the \emph{unary substructure} and the \emph{binary substructure}.
A \emph{unary substructure} $\usub$ is a union of all ground 1-types in $\structure$, and a \emph{binary substructure} $\bsub$ is a union of all ground 2-types in $\structure$.
A unary substructure and a binary substructure essentially represent a partition of a structure.
For instance, $\{R(e_1), \neg B(e_1), \neg E(e_1,e_1), B(e_2), \neg R(e_2), \neg E(e_2,e_2)\}$ and $\{\neg E(e_1, e_2), \neg E(e_2, e_1)\}$ are the unary substructure and binary substructure of the structure in \Cref{ex:types}, respectively, and they form a partition of the structure.
It means that enumerating a model can be achieved by enumerating its unary substructure and binary substructure consecutively.

Given an \fotwo{} sentence $\sentence$ with  1-types $\tau_1, \tau_2, \dots, \tau_{|U|}$, we define its \emph{1-type configuration} (or \emph{configuration} for brevity) $\vecn=(n_1, n_2, \dots, n_{|U|})$ as a vector of non-negative integers.
Each integer $n_i\ge 0$ represents the number of elements realizing $\tau_i$, and is also called the \emph{cardinality} of $\tau_i$.
Every model of $\sentence$ over a domain of size $|\vecn|$ corresponds to a unique configuration, though configurations may generally correspond to multiple models.
For instance, the configuration of the model in \Cref{ex:types} is $(1, 1)$.
By swapping the 1-types of the two elements in the model, we obtain a different model that shares the same configuration.
It is also possible for a configuration to have zero models, e.g., if we add the formula $\forall x \forall y \!:\! \neg (R(x) \land B(y))$ to the sentence in \Cref{ex:types}, the configuration $(1, 1)$ will no longer correspond to any model.
We call such configurations \emph{unsatisfiable}, otherwise, they are \emph{satisfiable}.

\section{Configuration Decision Problem}
\label{sec:foundational_lemmas}

In this section, we study the \emph{configuration decision problem}, which serves as a foundation for our enumeration algorithm.
We formally define it as follows:
\begin{problem}[Configuration Decision Problem]
  Fix an \fotwo{} sentence $\sentence$, 
  the \emph{configuration decision problem} $\confsat(\sentence, \vecn)$ asks whether a configuration $\vecn$ is satisfiable with respect to the sentence $\sentence$.
\end{problem}
As we will show in \Cref{sec:unary_substructure,sec:binary_substructure}, our enumeration algorithm incrementally extends a substructure to a complete model, and it is important to check whether the current substructure is consistent with the sentence.
The consistency check is essentially equivalent to a configuration decision problem as shown in \Cref{sec:from_substructure_to_configuration}.

The configuration decision problem is different from the classical decision problem for \fotwo{}, which asks whether there exists a model of a given sentence.
It is also different from the \emph{spectrum membership problem}~\cite{jones1972turing}, whose goal is to decide whether the input integer $n$ is in the spectrum of a given sentence, i.e., whether there exists a model of the sentence over a domain of size $n$.
In fact, there exists a hierarchy of reductions among these three problems:
due to the \emph{finite model property}~\cite{finitemodel} of \fotwo{}, the decision problem can be reduced to the spectrum membership problem, and the spectrum membership problem can be solved by the configuration decision problem (the domain size $n$ is in the spectrum if and only if there exists a satisfiable configuration $\vecn$ such that $|\vecn| = n$).

The configuration decision problem is more closely related to the spectrum membership problem, and our approach to solving the configuration decision problem is partially inspired by its solution.
There are two key properties of the spectrum membership problem that we generalize to the configuration decision problem for \fotwo : 
\begin{itemize}
  \item Monotonicity of Satisfiability (MS)~\cite{jones1972turing}: If a domain size $n$ is in the spectrum, then any domain size $n' > n$ is also in the spectrum;
  \item Exponential Size Model (ESM)~\cite{1997decision}: Every satisfiable \fotwo{} sentence has a model over a domain of size at most $3s2^r$, where $s$ is the size of the sentence and $r$ is the number of predicates in the sentence.
\end{itemize}
By these two properties, the spectrum membership problem can be solved in $\logO(1)$.
The idea is to first find the minimal domain size $n^*\le 3s2^r$ in the spectrum, if it exists, and then check whether the given domain size $n$ is less than or equal to $n^*$, otherwise directly output ``no''.

Our approach follows a similar approach, but uses generalized properties specific to the configuration decision problem.
More precisely, we define the \emph{derivation relation} as the partial order of configurations:
\begin{definition}[Derivation Relation]
  \label{def:configuration_derivation}
  Consider an \fotwo{} sentence $\sentence$ and its configuration $\vecn = (n_1, \dots, n_{|U|})$.
  We say that we can \emph{derive} a configuration $\vecn' = (n_1', \dots, n_{|U|}')$ from $\vecn$ if
  \begin{align*}
    \left\{
    \begin{aligned}
    &n_i' \ge n_i, \ \text{if } n_i > 0\\
    &n_i' = n_i, \ \text{otherwise}
    \end{aligned}
    \right.
  \end{align*}
  for all $i \in [|U|]$, and we denote this as $\vecn \preccurlyeq \vecn'$.
\end{definition}
Then we show that the MS as well as ESM properties also hold for the configuration decision problem in terms of the derivation relation, which is formalized as \cref{le:configuration_monotonicity} and \cref{le:upper_bound}, respectively.

\begin{restatable}{proposition}{configurationmonotonicity}
  \label{le:configuration_monotonicity}
  Consider an \fotwo{} sentence $\sentence$. 
  For any configuration $\vecn$ of $\sentence$, if $\vecn$ is satisfiable, then every configuration $\vecn'$ of $\sentence$ such that $\vecn \preccurlyeq \vecn'$ is also satisfiable.
\end{restatable}
We prove this proposition by induction.
For any satisfiable configuration $\vecn = (n_1, \dots, n_{|U|})$ and any configuration $\vecn' = (n_1, \dots, n_i + 1, \dots, n_{|U|})$ such that $n_i > 0$, we show that we can construct a model $\fomodel'$ of $\sentence$ with configuration $\vecn'$ from a model $\fomodel$ with configuration $\vecn$.
The construction is achieved by adding a new element $e$ realizing $\tau_i$ to the model $\fomodel'$ with all other elements and their 1-types and 2-types in $\fomodel$ intact.
The 2-types realized by $e$ and other elements are copied from another element realizing the same 1-type $\tau_i$ (which always exists by $n_i > 0$).
Then we show that $\fomodel'$ is also a model of $\sentence$.
The details of the proof are omitted here and can be found in the appendix.

We can also generalize the ESM property to the configuration decision problem.
Formalizing this, we use the notation $\vecn[i\mapsto k]$ to denote a new configuration that is identical to $\vecn$ except the $i$-th value replaced by $k$.
We then consider the \emph{universal bound} for the satisfiability of configurations.
\begin{restatable}{proposition}{upperbound}
  \label{le:upper_bound}
  Consider an \fotwo{} sentence $\sentence$. There is an upper bound 
  \begin{equation*}
    \delta=\max\{m(m+1), 2m+1\},
  \end{equation*}
  where $m$ is the number of formulas containing existential quantifiers in $\sentence$, such that for any configuration $\vecn = (n_1, \dots, n_{|U|})$ with $n_i > \delta$, if $\vecn$ is satisfiable, then $\vecn[i\mapsto \delta]$ is also satisfiable.
\end{restatable}

\begin{remark}
  Note that the size $|U|$ of a configuration is less than or equal to $2^{r}$, where $r$ is the number of predicates in the sentence.
  The above lemma implies that every satisfiable \fotwo{} sentence has a model of size at most $\delta \cdot 2^r$, providing an analogous ESM guarantee. 
\end{remark}

The basic proof idea of this proposition is similar to the proof of \cref{le:configuration_monotonicity}: We construct a model $\fomodel'$ of $\sentence$ with configuration $\vecn[i\mapsto \delta]$ from a model $\fomodel$ with configuration $\vecn$.
We only sketch the proof here and refer the reader to the appendix for the details.

For a given structure $\structure$, we say that an element $e$ is the \emph{(Skolem) witness} of an element $e'$ with respect to $\beta_k$ for some $k\in[m]$ if the literal $\beta_k(e',e)$ holds in $\structure$.
The element $e$ is also called a \emph{$\beta_k$-witness} of $e'$.
We define the \emph{necessary witness set} of an element $e$ in $\structure$ as a set of witnesses for $e$ that meets the following conditions:
\begin{itemize}
  \item \textbf{Coverage}: for each predicate $\beta_k$, there is at least one $\beta_k$-witness of $e$ in this set.
  \item \textbf{Minimality}: removing any witness from this set results in a set that no longer satisfies the coverage condition.
\end{itemize}
Every element in a model has at least one necessary witness set. 
The size of any necessary witness set for an element is at most $m$, since at most one witness is needed for each $\beta_k$.

Now we are ready to present the construction of a model $\fomodel'$ with configuration $\vecn[i\mapsto \delta]$.
We first copy $\fomodel$ to $\fomodel'$ with all elements realizing the 1-type $\tau_i$ removed.
Then we construct the set $C_i$ of elements realizing the 1-type $\tau_i$ in $\fomodel'$ such that every element in the domain of $\fomodel'$ has a necessary witness set, which means that $\fomodel'$ satisfies the existentially quantified formulas.
We prove that at most $m(m+1)$ elements in $C_i$ are sufficient for all other elements not in $C_i$ to have necessary witness sets.
Meanwhile, there are at most $2m+1$ elements needed to make the elements in $C_i$ also have necessary witness sets.
Thus the total number of elements in the constructed $C_i$ is at most $\delta = \max\{m(m+1), 2m+1\}$.
When constructing $C_i$, we ensure that the 2-types between elements are compatible with the respective 1-types.
Thus $\fomodel'$ is a model of $\sentence$ with configuration $\vecn[i\mapsto \delta]$, leading to the conclusion.

Combining \cref{le:configuration_monotonicity,le:upper_bound}, we have the following theorem, whose proof follows directly from Lemma~\ref{le:configuration_monotonicity} and Lemma~\ref{le:upper_bound}.
\begin{theorem}
  \label{th:1-type_configuration_satisfiability}
  Let $\sentence$ be an \fotwo{} sentence.
  A configuration $\vecn$ is satisfiable if and only if there exists a satisfiable configuration $\vecn' \in  \{0, 1, \dots, \delta\}^{|U|}$ such that $\vecn' \preccurlyeq \vecn$, where $\delta$ is the upper bound defined in Proposition~\ref{le:upper_bound}.
\end{theorem}

Recall that the \fotwo{} sentence $\sentence$ is fixed.
We can iterate over configurations in $\{0, 1, \dots, \delta\}^{|U|}$ and find all satisfiable configurations.
Note that the number of configurations in $\{0, 1, \dots, \delta\}^{|U|}$ is independent of the domain size $n$, and so is the number of satisfiable configurations.
Then according to \cref{th:1-type_configuration_satisfiability}, we can determine the satisfiability of arbitrary configuration $\vecn$ by simply checking whether it can be derived from one of these satisfiable configurations.
Therefore, we have the following corollary:
\begin{corollary}
  \label{cor:configuration_decision}
  The configuration decision problem $\confsat(\sentence, \vecn)$ can be solved in $\logO(1)$.
\end{corollary}

\section{Unary Substructure Enumeration}
\label{sec:unary_substructure}

As we discussed when introducing the substructures, one can enumerate models in two phases:
\begin{itemize}
  \item Unary Substructure Enumeration: Enumerate unary substructures $\usub$ such that $\usub$ is $\domain$-consistent with $\sentence$;
  \item Binary Substructure Enumeration: Conditioned on the unary substructure $\usub$, enumerate binary substructures $\bsub$ such that $\usub \cup \bsub$ is a model of $\sentence$.
\end{itemize}
In this section, we focus on the first phase.

Enumerating unary substructures is equivalent to enumerating 1-type assignments for elements in the domain.
We achieve this by first enumerating satisfiable configurations, and then enumerating 1-type assignments according to the configurations.

\subsection{Satisfiable Configuration Enumeration}
\label{sub:satisfiable_configuration}

From~\cref{cor:configuration_decision}, we know that the configuration decision problem can be solved efficiently in $\logO(1)$.
Intuitively, given a domain size $n$, we can enumerate satisfiable configurations by iterating over all possible configurations (i.e., the partitions of $n$ to $|U|$ parts) and checking whether they are satisfiable.
However, this would eventually lead to a delay complexity exceeding $\logO(n^2)$.

Observe that by \cref{th:1-type_configuration_satisfiability}, each satisfiable configuration can be derived from a configuration whose cardinalities are bounded by $\delta$, which provides another clue.
We can first obtain these bounded satisfiable configurations and then derive further satisfiable configurations from them.
We refer to such bounded satisfiable configurations as \emph{template configurations}.
\begin{definition}[Template Configuration]
  \label{def:template_configuration}
  Given an \fotwo{} sentence $\sentence$, if a configuration  $\bar{\vecn} \in  \{0, 1, \dots, \delta\}^{|U|}$ is satisfiable, then $\bar{\vecn}$ is a \emph{template configuration} of $\sentence$.
\end{definition}
Note that the number of template configurations is also independent of the input domain size $n$.
Thus, we can obtain all template configurations beforehand with a constant complexity.
Now we consider how to enumerate satisfiable configurations $\vecn$ with $|\vecn| = n$ based on the template configurations.
For a template configuration $\bar{\vecn}$, we use $N_{\delta}$ to denote the number of 1-types whose cardinalities are equal to $\delta$ in the configuration.
The core idea is to iterate template configurations and for each template configuration $\bar{\vecn}$, partition the integer $n - |\bar{\vecn}|$ into $N_{\delta}$ parts, and add them to respective cardinalities whose value is equal to $\delta$.
According to Lemma \ref{le:configuration_monotonicity}, the new configurations are all satisfiable (they are derived from a satisfiable configuration).
The following example illustrates the process of enumerating satisfiable configurations from a template configuration.
\begin{example}
  Let $\delta = 6$ and $n = 21$.
  Consider a template configuration $(6,4,6,2)$.
  There are $18$ elements in total in the template configuration and $N_{\delta} = 2$.
  So we partition $3$ into $2$ parts, i.e., $(3, 0)$, $(2, 1)$, $(1, 2)$ and $(0, 3)$, and add them to the first and third cardinalities in the template configuration.
  Figure~\ref{fig:1-type_config} shows the illustration of this process.
\end{example}
\begin{figure}
  \centering
  \includegraphics[width=0.48\textwidth]{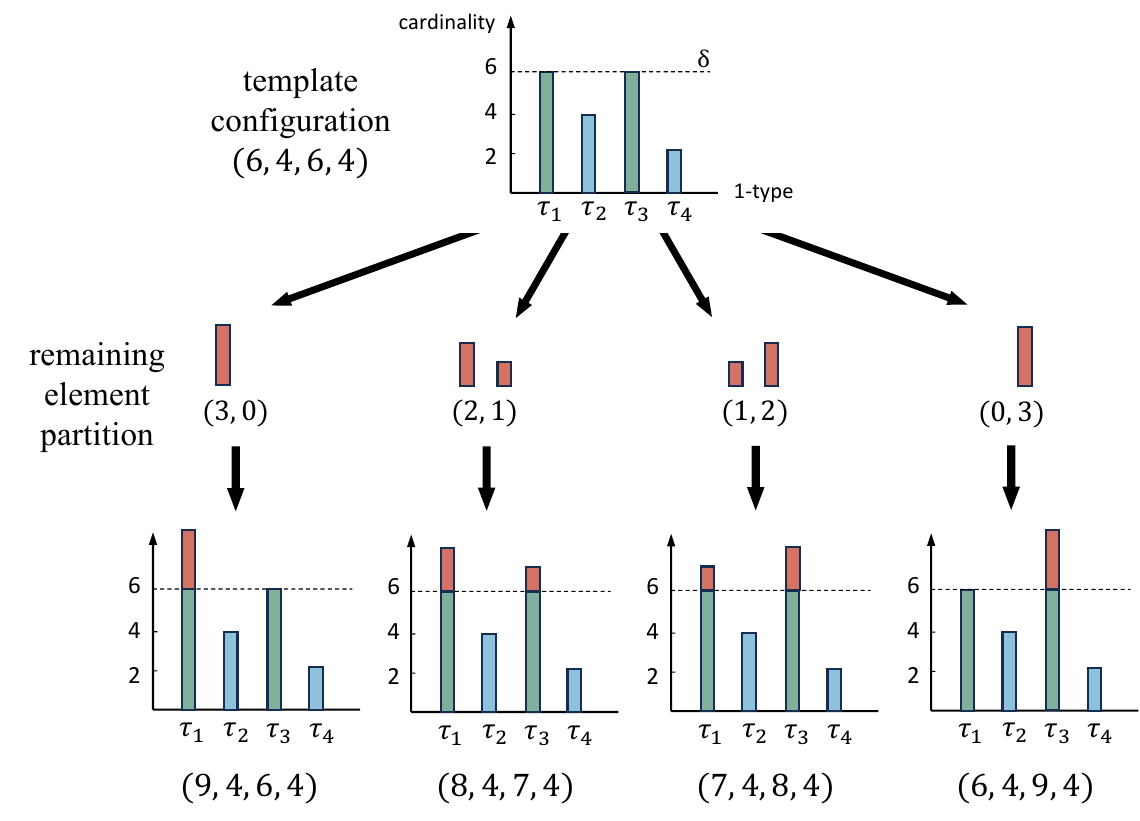}
  \caption{An example of enumerating satisfiable configurations from a template configuration $(6,4,6,4)$, where $\delta = 6$ and $n = 21$.}
  \label{fig:1-type_config}
\end{figure}
In addition to the base case shown in the previous example, there are some special cases that we need to further consider.
When the given domain size $n$ is small enough, the number of elements in some template configurations may exceed $n$ (e.g., $n < 18$ in the previous example).
In this case, we skip them since we cannot derive a proper configuration.
On the other hand, there may be also no 1-type whose cardinality is equal to $\delta$ in some template configurations (i.e., $N_{\delta} = 0$), and thus the remaining number of elements cannot be partitioned.
We also directly skip in this case.

In summary, given a domain size $n$, we enumerate satisfiable configurations by iterating over all template configurations with the following steps: For each template configuration $\bar{\vecn}$, 
\begin{enumerate}
  \item if $|\bar{\vecn}| > n$, skip it,
  \item if $|\bar{\vecn}| = n$, it is a satisfiable configuration in itself, and we directly yield it,
  \item if $|\bar{\vecn}| < n$ and $N_{\delta} = 0$, skip it, and
  \item if $|\bar{\vecn}| < n$ and $N_{\delta} > 0$, partition the number of remaining elements into $N_\delta$, add them to $\bar{\vecn}$, as described previously, and yield the resulting configuration.
\end{enumerate}

\begin{lemma}
  Fix an \fotwo{} sentence $\sentence$, for a positive integer $n$, the above method can enumerate all satisfiable configurations $\vecn$ of $\sentence$ such that $|\vecn| = n$ with a delay complexity $\logO(1)$.
  \label{le: enumerate configs}
\end{lemma}
\begin{proof}
  First, we show that the method can enumerate all satisfiable configurations.
  It is obvious that the configurations obtained by the method are all satisfiable and the sum of their cardinalities is equal to $n$.
  Furthermore, the method can obtain all satisfiable configurations: For any satisfiable configuration $\vecn$ such that $|\vecn| = n$, 
  \begin{itemize}
    \item if $\vecn$ does not contain any cardinality larger than $\delta$, then $\vecn$ is a template configuration and can be obtained by the case $2$, and
    \item otherwise, we can always find a template configuration $\bar{\vecn}$ with the same cardinalities as $\vecn$ except for those larger than $\delta$, and obtain $\vecn$ by partitioning the number of remaining elements as described in the case $4$.
  \end{itemize}
  Next, we show that the satisfiable configurations yielded by the method are all distinct.
  \begin{itemize}
    \item If the configuration does not contain any cardinality larger than $\delta$, then it is a template configuration, and thus must be distinct (it can only be obtained by the case $2$);
    \item If there is a cardinality larger than $\delta$ in the configuration, then it must be obtained by the case $4$.
    Since the template configuration iterated in the case $4$ is distinct, combined with the partition process, we know that the configurations obtained by the case $4$ are also distinct.
  \end{itemize}

  Finally, we show that the delay complexity is $\logO(1)$.
  Since the number of template configurations is independent of domain size $n$, even in the worst case, the number of configurations we skip (due to the case $1$ and case $3$) is independent of $n$.
  For the case $4$, as $N_\delta$ is also independent of $n$, enumerating partitions of the integer $n-|\vecn|$ can be done with a constant time delay complexity.
\end{proof}

\subsection{Unary Substructure Enumeration Algorithm}

We consider the domain $\domain = \{e_1, \dots, e_n\}$ of size $n$.
For a satisfiable configuration $\vecn=(n_1, \dots, n_{|U|})$ such that $|\vecn| = n$, we can enumerate all its corresponding unary substructures over $\domain$ by partitioning the domain according to $\vecn$.
For example, given a configuration $\vecn = (3,2,1)$, we can partition the domain into three disjoint subsets $\{e_1, e_2, e_3\}$, $\{e_4, e_5\}$, and $\{e_6\}$, and assign the first, second, and third 1-types in $U$ to the elements in these three subsets, respectively.
By the definition of configurations satisfiability and the consistency of unary substructures, we know that the enumerated unary substructures are all $\domain$-consistent with the given sentence $\sentence$.

Enumerating all possible partitions of the domain according to $\vecn$ is a classic combinatorial problem, and there are plenty of efficient algorithms~\cite{goulden2004combinatorial}.
When $|U|$ is fixed (due to the fixed sentence $\sentence$), the delay complexity of enumerating partitions of a configuration is $\logO(n)$.

The full algorithm for enumerating unary substructures is shown in Algorithm~\ref{alg:unary_substructure}.
The \emph{``yield''} operator is used to output values to its calling context.
Intuitively, when we call a function that contains a ``yield'' statement, the state of the function is preserved while yielding and resuming a suspended iterator causes it to continue from the point of yielding.
It is very useful for enumeration algorithms, where we do not need to handle the current state of the enumeration process explicitly.
Due to the constant complexity assumption of outputting a value, the yielding process can be done in linear time. 
The function $\mathsf{EnumSatConfigs}(\sentence, n$) yields satisfiable configurations $\vecn$ of $\sentence$ such that $|\vecn| = n$ by the method in the previous subsection, and the function $\mathsf{EnumPartitions}(\domain,\vecn)$ produces unary substructures over $\domain$ by partitioning $\domain$ according to the configuration $\vecn$.
\begin{algorithm}[tbp]
  \caption{$\mathsf{EnumUnarySubstructures}(\sentence, n)$}
  \label{alg:unary_substructure}
  \textbf{Input:} An \fotwo{} sentence $\sentence$ and a positive integer $n$ \\
  \textbf{Output:}  Unary substructures of $\sentence$ over $\domain = \{e_1, \dots, e_n\}$
  \begin{algorithmic}[1]
    \State $\domain \gets \{e_1, \dots, e_n\}$
    \For{$\vecn$ $\gets$ $\mathsf{EnumSatConfigs} (\sentence, n$)}
      \For{$\usub \gets$ $\mathsf{EnumPartitions} (\domain, \vecn)$}
        \State \textbf{yield} $\usub$ 
      \EndFor
    \EndFor
  \end{algorithmic}
\end{algorithm}

\begin{lemma}
  Fix an \fotwo{} sentence $\sentence$, for a positive integer $n$, the delay complexity of enumerating unary substructures of $\sentence$ over a domain of size $n$ is $\logO(n)$.
\end{lemma}
\begin{proof}
  The proof follows directly from \cref{alg:unary_substructure} and \cref{le: enumerate configs}.
\end{proof}

\section{Binary Substructure Enumeration}
\label{sec:binary_substructure}

In this section, we focus on the second phase of the enumeration algorithm, enumerating binary substructures conditioned on a unary substructure.
The core idea is \emph{recursion and backtracking}.
Specifically, we construct a binary substructure by recursively determining the 2-types between each pair of elements, and exploring all possible binary substructures via backtracking.

\subsection{A Working Example}

Let us first introduce a working example that will be used throughout this subsection.
Consider an \fotwo{} sentence $\sentence_G$ that describes undirected graphs without any isolated vertex:
\begin{equation}
  \label{eq:example_sentence}
  \begin{aligned}
      \sentence_G = \ & \forall x \forall y: \neg E(x,x) \land (E(x,y) \rightarrow E(y,x)) \ \land \\
      &\forall x \exists y: E(x,y)
  \end{aligned}
\end{equation}
and a domain size $n = 3$. Let the domain $\domain = \{v_1,v_2,v_3\}$.

There is only one compatible 1-type $\tau(x) = \neg E(x,x)$ of $\sentence$, and two 2-types that are compatible with $\tau(x)$:
\begin{equation}
  \begin{gathered}
    \pi_E(x,y) = E(x,y) \land E(y,x), \\
    \pi_{\neg E}(x,y) = \neg E(x,y) \land \neg E(y,x).
  \end{gathered}
  \label{eq:example_types}
\end{equation}
These 2-types represent whether two elements $v_i$ and $v_j$ are connected ($\pi_E(v_i,v_j)$ holds) or not connected ($\pi_{\neg E}(v_i,v_j)$ holds) in the graph.
There is only one $\domain$-consistent unary substructure $\usub = \{\neg E(v_1, v_1), \neg E(v_2, v_2), \neg E(v_3, v_3)\}$.

\begin{figure*}
  \centering
  \includegraphics[width=0.9\textwidth]{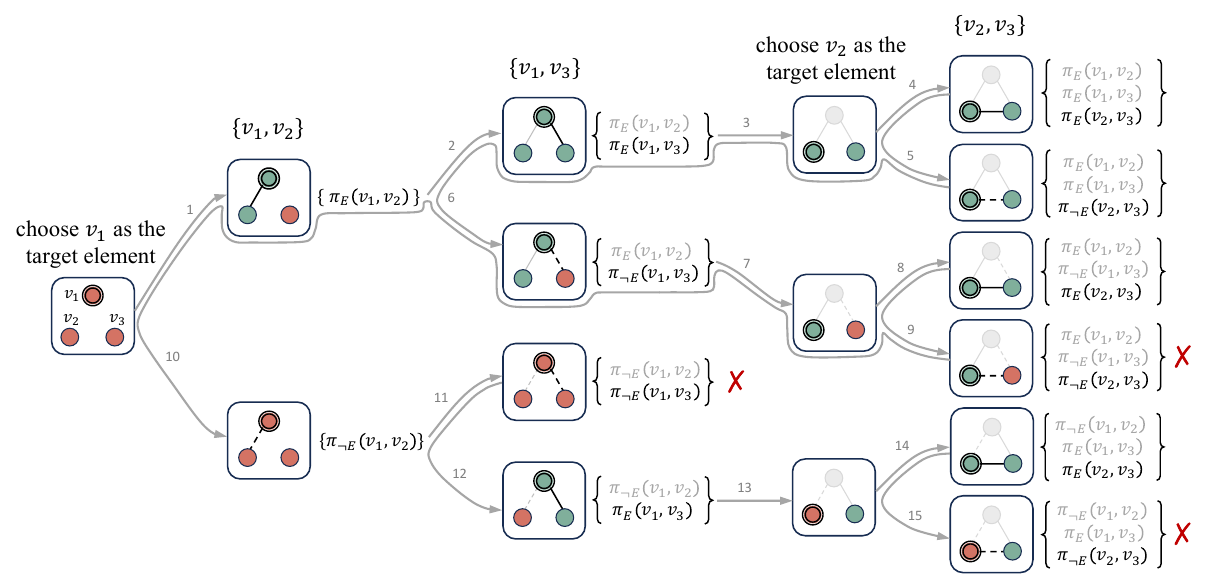}
  \caption{The process of model enumeration of the sentence $\sentence_G$ over the domain $\domain = \{v_1,v_2,v_3\}$. 
  Arrows between graphs and their numbers show the recursion and backtracking path. 
  The solid and dashed edges in graphs represent the 2-types $\pi_E$ and $\pi_{\neg E}$ respectively, and the green and red represent non-isolated and isolated vertices respectively. 
  The double circle represents the chosen target vertex. The sets next to the graphs indicate the partial binary substructure determined so far.
  The cross indicates that the substructure of this step is not $\domain$-consistent. 
  We mark the determined 2-types in gray, and also mark the vertex whose 2-types have been determined in gray, indicating that they will not change in subsequent processing. 
  When all 2-types of vertex pairs are processed, we yield a binary substructure, as shown in the last column.}
  \label{fig:example_tree}
\end{figure*}

We briefly describe the binary substructure enumeration process for this example as illustrated in Figure~\ref{fig:example_tree}.
We begin by initializing a binary substructure $\bsub$ as an empty set and choosing an arbitrary element as the \emph{target element}. 
Without loss of generality, let this element be $v_1$.
We determine the 2-type of the element pairs involving $v_1$, i.e., $\{v_1, v_2\}$ and $\{v_1, v_3\}$, one by one.
For each pair, we iterate over all possible 2-types, i.e., $\pi_E$ and $\pi_{\neg E}$.
Consider the pair $\{v_1, v_2\}$ and its 2-type $\pi_E$, if the substructure $\usub \cup \bsub \cup \pi_E(v_1,v_2)$ is $\domain$-consistent with $\sentence_G$, we add $\pi_E(v_1,v_2)$ to the binary substructure $\bsub$ we are constructing and continue to process the next pair $\{v_1, v_3\}$.
Once all 2-types involving $v_1$ have been determined, the $v_1$-related literals (including its 1-type and all 2-types with other elements) are fully determined. 
We then choose another element $v_2$ as the target element and repeat the same process.
After 2-types of all pairs are determined, we yield a complete binary substructure $\bsub$.
At this point, we backtrack to the previous step and explore different 2-types to enumerate other binary substructures.
If a substructure is not $\domain$-consistent with $\sentence_G$ after adding a ground 2-type, we directly backtrack to the previous step and try a different 2-type.
For example, in Step $11$ of Figure~\ref{fig:example_tree}, specifying $\pi_{\neg E}$ as the 2-type of $\{v_1, v_3\}$ leads to an inconsistent substructure ($v_1$ is an isolated vertex).
Consequently, we backtrack to the previous step and try another 2-type $\pi_E$ in step $12$.


\subsection{Recursion and Backtracking Algorithm}

Now we present the method of enumerating binary substructures conditioned on a given unary substructure $\usub$.
The domain considered is still $\domain = \{e_1, \dots, e_n\}$, and we leave out $\domain$ in ``$\domain$-consistent'' when the context is clear.

As we discussed, the process of constructing a binary substructure involves assigning 2-types to all pairs of elements.
We process these pairs by the recursion on elements: select a target element, process all pairs related to it, and then move on to the next element, repeating until all pairs are processed.
Note that the order of our recursion is not a traditional depth-first search that traverses all pairs along a single path.
It is instead a \emph{domain-based recursion}.\footnote{This resembles the domain recursion used in first-order model sampling~\cite{wang_exact_2023-1,wang2024lifted}.
The key difference is that in enumeration, only one 2-type is determined at each step, while in sampling, all 2-types related to an element are sampled at once.}
Then, all binary substructures are enumerated by backtracking.
During each 2-type assignment, we check whether the current substructure is consistent with the sentence $\sentence$ to avoid unnecessary recursion.
Formally, it is to solve a \emph{substructure decision problem}.
\begin{problem}[Substructure Decision Problem]
  Fix an \fotwo{} sentence $\sentence$, 
  the \emph{substructure decision problem} $\subsat(\sentence, \domain, \substr)$ asks whether the substructure $\substr$ is $\domain$-consistent with the sentence $\sentence$.
\end{problem}

\begin{algorithm}[tbp]
  \caption{Model Enumeration for \fotwo{}}
  \label{alg:recursive}
  \textbf{Input:} An \fotwo{} sentence $\sentence$, a positive integer $n$ \\
  \textbf{Output:} Models of the sentence $\sentence$ over $\domain = \{e_1, \dots, e_n\}$
  \begin{algorithmic}[1]
    \State Set global variables $\sentence$ and $\domain \gets \{e_1, \dots, e_n\}$
    \For {$\usub$ $\gets$ $\mathsf{EnumUnarySubstructures}(\sentence, n$)}
      \For {$\fomodel \gets \mathsf{DomainRecursion}(\domain, \usub)$}
        \State \textbf{yield} $\fomodel$
      \EndFor
    \EndFor

    \State

    \Function{}{}$\mathsf{DomainRecursion}(\domain^{\bot}, \substr)$
      \If {$|\domain^{\bot}| = 1$} \label{line:base_case}
        \State \textbf{yield} $\substr$ \Comment{$\substr$ is actually a model}
        \State \Return
      \EndIf
      \State $\targete \gets$ Choose a target element from $\domain^{\bot}$ \label{line:choose_element}
      \For {$\substr \gets \mathsf{PairRecursion}(\targete, \ \domain^{\bot} \! \setminus \! \{\targete\}, \ \substr)$} \label{line:local_struct_sample}
        \State $\mathsf{DomainRecursion}(\domain^{\bot} \! \setminus \! \{\targete\}, \ \substr)$ \label{line:rec_call_domain_recursion}
      \EndFor
    \EndFunction

    \State
    
    \Function{}{}$\mathsf{PairRecursion}(\targete, \domain^{\!-}, \substr$)
      \If {$\domain^-$ is $\varnothing$}
        \State \textbf{yield} $\substr$ \Comment{$\substr$ contains a local substructure of $\targete$} \label{line:yield_local_struct}
        \State \Return
      \EndIf
      \State $e \gets$ Choose an element from $\domain^{\!-}$ \label{line:choose_element_2}
      \For {each compatible 2-type $\pi$ of pair $\{e^*, e\}$} \label{line:iterate_2types}
        \If {$\subsat(\sentence, \domain, \substr\cup \pi(\targete,e))$} \label{line:sat_check}
          \State $\mathsf{PairRecursion}(\targete, \domain^{\!-} \! \setminus \! \{e\}, \substr\cup\pi(\targete,e))$ \label{line:recursive_inner}
        \EndIf
      \EndFor
    \EndFunction
  \end{algorithmic}
\end{algorithm}

Our algorithm is shown in Algorithm~\ref{alg:recursive}.
For ease of description, we use the term \emph{local substructure} to denote the substructure of an element $e$ that consists of the ground 2-types between $e$ and all other elements, i.e., $\bigcup_{e' \in \domain \setminus \{e\}} \pi(e, e')$.
The algorithm starts with the unary substructure $\usub$ enumerated by \cref{alg:unary_substructure}, and employs a recursion mechanism with two levels for enumerating the local substructure of each element.
\begin{itemize}
  \item The outer recursion via the function $\mathsf{DomainRecursion}$: 
  The parameter $\domain^{\bot}$ consists of elements, for which the algorithm has not started to process their local substructures yet.
  Choose an element $\targete$ from $\domain^{\bot}$, termed the \emph{target element}, and invoke the inner recursion $\mathsf{PairRecursion}$ to obtain its local substructure (\cref{line:local_struct_sample}).
  Each time yielding a local substructure, we recursively call $\mathsf{DomainRecursion}$ to process the remaining elements in $\domain^{\bot} \setminus \{\targete\}$ (\cref{line:rec_call_domain_recursion}).
  When the outer recursion reaches its base case (\cref{line:base_case}), where there is only one element left to process, the structure $\substr$, which is actually a model, is yielded, and then the algorithm backtracks into the previous step of calling $\mathsf{PairRecursion}$ in the for-loop (\cref{line:local_struct_sample}) (note that the keyword ``return'' interrupts the resuming by the yield statement).
  \item The inner recursion via the function $\mathsf{PairRecursion}$: 
  The parameter $\domain^-$ refers to the set of elements whose 2-types with $\targete$ have not been determined (i.e., their 2-types are not included in the local substructure of $\targete$).
  $\substr$ contains the unary substructure as well as the partial binary substructure constructed so far.
  In this function, we choose an element $e$ from $\domain^-$, and then iterate every compatible 2-types $\pi$ with the realized 1-types of $\targete$ and $e$.
  If $\substr \cup \pi(\targete,e)$ is consistent with $\sentence$ (by solving a substructure decision problem), add $\pi(\targete,e)$ to $\substr$ and call the inner recursion to process the next pair related to the target element (\cref{line:recursive_inner}).
  Since $S$ is consistent with $\sentence$, we can always find a 2-type $\pi$ passing the consistency check.
  Once all 2-types related to $\targete$ are determined, we obtain a local substructure of $\targete$, and the algorithm continues to process the outer recursion for the next element (\cref{line:yield_local_struct}).
\end{itemize}


Since throughout the algorithm, we always ensure that the substructure is consistent with the sentence $\sentence$, the final enumerated structure is a model of $\sentence$ over $\domain$.
It is easy to check that the algorithm enumerates all models without duplicates.

\begin{proposition}\label{lem:enumeration}
  \cref{alg:recursive} correctly enumerates all models of an \fotwo{} sentence $\sentence$ over a domain $\domain$ of size $n$.
\end{proposition}
\begin{proof}
  For any model \fomodel of $\sentence$ over $\domain$, its unary substructure $\usub$ must be enumerated by $\mathsf{EnumUnarySubstructures}$ without duplicates according to \cref{sec:unary_substructure}.
  Without loss of generality, we assume that the algorithm processes the elements in the order $e_1, e_2, \dots, e_n$ in Lines~\ref{line:choose_element} and~\ref{line:choose_element_2}.
  Let $\pi_{ij}$ be the 2-type realized by the pair $\{e_i, e_j\}$ in \fomodel, and $\bsub^{i} = \bigcup_{j=i+1}^n \pi_{ij}(e_i,e_j)$.
  Then, it can be proved by induction that, for any $i\in[n]$ and $j=i+1,\dots,n$, the substructures 
  $$
  \bsub^{1}\cup \bsub^{2} \cup \dots \cup \bsub^{i-1}\cup \pi_{i,i+1}(e_i,e_{i+1})\cup \dots \cup \pi_{i,j}(e_i,e_j)$$ 
  can be enumerated by the algorithm.
  Specifically, these substructures are processed as $\substr$ in $\mathsf{PairRecursion}$.
  The non-duplication of models is guaranteed by the iteration over 2-types at \cref{line:iterate_2types}.
\end{proof}

\cref{alg:recursive} implicitly utilizes an early stopping mechanism to avoid exploring inconsistent substructures.
When adding a ground 2-type to the substructure $\substr$ will lead to its inconsistency (\cref{line:sat_check}), the algorithm will not proceed further and try another 2-type.
This ensures that every recursive call of $\mathsf{PairRecursion}$ produces a 2-type that must be a part of a model.
With this mechanism, one can easily check that the enumeration of any model only involves $\logO(n^2)$ times of $\mathsf{PairRecursion}$.

Now, the remaining problem is how to efficiently check the consistency of a substructure $\substr$, i.e., how to solve the substructure decision problem $\subsat(\sentence, \domain, \substr)$.
We show that this problem can be reduced to a configuration decision problem in the next section.

\section{From Substructure Decision Problem to Configuration Decision Problem}
\label{sec:from_substructure_to_configuration}

In this section, we illustrate how to reduce a substructure decision problem to a configuration decision problem. 

Recall that every substructure $\substr$ in \cref{alg:recursive} that needs to be checked for consistency with the sentence $\sentence$ (\cref{line:sat_check})
 has the form as stated in \cref{lem:enumeration}:
\begin{equation}
  \label{eq:substructure_format}
  \bsub^{1} \cup \dots \cup \bsub^{i-1}\cup \pi_{i,i+1}(e_i,e_{i+1})\cup \dots \cup \pi_{i,j}(e_i,e_j).
\end{equation}
Based on this, we first classify the elements in the domain $\domain$ for further processing.
The elements $\{e_1, \dots, e_{i-1}\}$, whose local substructures have been determined, denoted as $\domain^{\! \top}$, and the remaining elements, denoted as $\domain^\bot$.
The set $\domain^\bot$ is further partitioned into three disjoint parts: the target element $\targete$ (i.e., $e_i$ in \cref{eq:substructure_format}), the elements whose 2-types with $\targete$ are determined, denoted as $\domain^+$ (i.e., $\{e_{i+1}, ..., e_{j}\}$ in \cref{eq:substructure_format}), and the remaining elements, denoted as $\domain^-$.
These categories and symbols are also consistent with those in \cref{alg:recursive}.
The hierarchical partition can be shown as follows:
\begin{displaymath}
  \domain\ \left\{ \begin{array}{ll}
  \domain^{\! \top} \\
  \domain^\bot \
    \left\{ \begin{array}{ll}
      \{\targete\} \\
      \domain^+ \\
      \domain^- 
    \end{array} \right.
  \end{array} \right.
  \label{eq:domain_partition}
\end{displaymath}

\begin{example}
  Consider the example of enumerating non-isolated graphs, 
  after the third step in~\cref{fig:example_tree}, the target element $\targete$ is $v_2$, $\domain^+$ is empty, $\domain^-$ is $\{v_3\}$, and $\domain^{\! \top}$ is $\{v_1\}$.
\end{example}


\subsection{Domain Reduction}
According to \cref{alg:recursive}, the elements in $\domain^{\! \top}$ have their local substructures fully determined and satisfy the sentence $\sentence$.
Thus, they might be ignored when checking the consistency of the substructure $\substr$.
To safely remove them from $\domain$, we must consider their effect on the remaining elements in $\domain^\bot$.

We need the concept of \emph{$\beta$-satisfication}.
Recall that the sentence $\sentence$ is in SNF:
$$
  \sentence = \forall x \forall y: \phi(x,y) \land \bigwedge_{k\in[m]} \forall x\exists y: \beta_k(x,y).
$$
Given a substructure $\substr$, we say that an element $e$ is \emph{$\beta_k$-satisfied} in $\substr$ if there exists a positive literal $\beta_k(e,e')$ in $\substr$ for some element $e'$.
A model $\fomodel$ of $\sentence$ must satisfy that all elements in $\domain$ are $\beta_k$-satisfied in $\fomodel$ for all $k\in[m]$.

Let $\substr_{\domain^{\! \top}}$ be the union of all local substructures and ground 1-types of $\domain^{\! \top}$ in $\substr$.
By the consistency checks throughout the algorithm, we know that all elements $e\in\domain^{\! \top}$ are $\beta_k$-satisfied in $\substr_{\domain^{\! \top}}$ for all $k\in[m]$.
To encode the $\beta$-satisfactions of the elements in $\domain^{\! \bot}$, we introduce a set of Tseitin predicates $\{Z_1/1,Z_2/1,\dots,Z_m/1\}$ and transform $\sentence$ into an auxiliary sentence $\sentence'$:
\begin{align}
  \sentence' = & \ \forall x \forall y: \phi(x,y) \ \land \label{eq:enum_aux_formula_phi} \\ 
  &\bigwedge_{k\in[m]} \forall x: Z_k(x) \rightarrow (\exists y: \beta_k(x,y)). \label{eq:enum_aux_formula_Z}
\end{align}

Using the Tseitin predicates, we can capture the effect of the local substructures of $\domain^{\! \top}$ on the remaining elements in $\domain^\bot$ by manipulating $\substr$.
For any $e\in \domain^{\! \bot}$, denote by $\mathcal{Z}_e$ the interpretation of the Tseitin predicates over $e$ such that $Z_k(e)$ is \emph{negative} if and only if $e$ is $\beta_k$-satisfied in $\substr_{\domain^{\! \top}}$.
Let $\mathcal{Z} = \cup_{e\in \domain^{\bot}} \mathcal{Z}_e$ and let
\begin{equation*}
  \substr' = \substr \setminus \substr_{\domain^{\! \top}} \cup \mathcal{Z}.
\end{equation*}

\begin{example}
  \label{ex:domain_reduction}
  Consider the previous non-isolated graph example, the auxiliary sentence $\sentence'_G$ is as follows:
  \begin{align*}
    \tsentence_G = \ &\forall x \forall y: \neg E(x,x) \land (E(x,y) \rightarrow E(y,x)) \ \land \\ 
    &\forall x: Z(x) \rightarrow (\exists y: E(x,y)). 
  \end{align*}
  After the seven step in~\cref{fig:example_tree}, the resulting substructure $\substr'$ is $\{\neg E(v_2,v_2), \neg E(v_3,v_3), \neg Z(v_2), Z(v_3)\}$.
\end{example}

The original substructure decision problem can be reduced to another one of $\sentence'$ over the reduced domain $\domain^{\bot}$ with the new substructure $\substr'$:
\begin{lemma}
  \label{le:domain_reduction}
  Any $\subsat(\sentence, \domain, \substr)$ appearing in \cref{alg:recursive} is equivalent to the problem $\subsat(\sentence', \domain^{\bot}, \substr')$.
\end{lemma}
\begin{proof}
  ($\Rightarrow$) Suppose that $\subsat(\sentence, \domain, \substr)$ holds, and let $\fomodel_\sentence$ be any model of $\sentence$ over $\domain$ such that $\substr \subseteq \fomodel$.
  We construct a structure $\fomodel_{\sentence'}$ of $\sentence'$ over $\domain^{\! \bot}$:
  \begin{equation*}
    \fomodel_{\sentence'} = \fomodel_{\sentence} \setminus \substr_{\domain^{\! \top}} \cup \mathcal{Z}.
  \end{equation*}
  It is easy to check that $\substr'\subseteq \fomodel_{\sentence'}$.
  We show that $\fomodel_{\sentence'}\models \sentence'$.
  Clearly, $\fomodel_{\sentence'} \models \phi(x,y)$ and we only consider the existential part \cref{eq:enum_aux_formula_Z}.
  For any element $e\in \domain^{\! \bot}$, if $Z_k(e) \in \mathcal{Z}_e$, which means that $e$ is not $\beta_k$-satisfied in $\substr_{\domain^{\! \top}}$, then there exists a positive literal $\beta_k(e,e')$ in $\fomodel_{\sentence}\setminus\substr_{\domain^{\! \top}}$ for some $e'\in\domain^\bot$ (otherwise $\fomodel_{\sentence}$ would not be a model of $\sentence$), and thus $\exists y: \beta_k(e,y)$ holds in $\fomodel_{\sentence'}$.
  Otherwise, if $Z_k(e) \notin \mathcal{Z}_e$, then the implication in \cref{eq:enum_aux_formula_Z} is trivially satisfied.
  Therefore, $\fomodel_{\sentence'} \models \sentence'$, and hence $\subsat(\sentence', \domain^\bot, \substr')$ holds.

  ($\Leftarrow$) Conversely, suppose that $\subsat(\sentence', \domain^{\bot}, \substr')$ holds, and let $\fomodel_{\sentence'}$ be any model of $\sentence'$ over $\domain^{\bot}$ such that $\substr' \subseteq \fomodel_{\sentence'}$.
  We construct a structure of $\sentence$ over $\domain$:
  \begin{equation*}
    \fomodel_{\sentence} = \fomodel_{\sentence'} \cup \substr_{\domain^{\! \top}} \setminus \mathcal{Z}.
  \end{equation*}
  We show that $\fomodel_{\sentence}\models \sentence$.
  Similarly, $\fomodel_{\sentence} \models \phi(x,y)$ and we only need to consider the existential part, that is, whether all elements $e \in \domain$ are $\beta_k$-satisfied in $\fomodel_{\sentence}$ for all $k\in[m]$.
  The $\beta$-satisfactions of elements in $\domain^{\! \top}$ are already guaranteed by the construction of $\fomodel_{\sentence}$.
  For $e\in \domain^{\! \bot}$, there are two cases:
  \begin{itemize}
    \item If $Z_k(e) \in \mathcal{Z}_e$, there must exist a positive literal $\beta_k(e,e')$ in $\fomodel_{\sentence'}$ for some $e'$ (otherwise $\fomodel_{\sentence'}$ would not be a model of $\sentence'$), and thus $e$ is $\beta_k$-satisfied in $\fomodel_{\sentence}$;
    \item If $\neg Z_k(e) \in \mathcal{Z}_e$, there must exist a positive literal $\beta_k(e,e')$ in $\substr_{\domain^{\! \top}}$ for some $e'$ (by the definition of $\mathcal{Z}_e$), and thus $e$ is also $\beta_k$-satisfied in $\fomodel_{\sentence}$.
  \end{itemize}
  Therefore, $\fomodel_{\sentence} \models \sentence$. Finally, it is easy to check that $\substr \subseteq \fomodel_{\sentence}$, and thus $\subsat(\sentence, \domain, \substr)$ holds.
\end{proof}

\subsection{Unary Encoding of Binary Facts}
After reducing the domain $\domain$ to $\domain^{\bot}$ and converting $\substr$ to $\substr'$ as described above, all remaining binary literals in $\substr'$ are between the target element $\targete$ and the elements in $\domain^+$.
We now aim to represent these binary literals in $\substr'$ using some unary literals, thereby transforming $\substr'$ into a unary substructure.

To achieve this, we introduce several auxiliary unary predicates: the \emph{target predicate} $T/1$ to identify the target element, and a set of \emph{relation predicates} $\{R_\pi/1 \mid \pi \in B_\sentence\}$ to encode the 2-types between the target element and the elements in $\domain^+$.
Recall that $B_\sentence$ denotes the set of all 2-types in the sentence $\sentence$.
We then transform the sentence $\sentence'$ (\cref{eq:enum_aux_formula_phi,eq:enum_aux_formula_Z}) to another auxiliary sentence $\tsentence$:
\begin{align}
  \tsentence = 
  & \ \sentence' \land \bigwedge_{\pi\in B_\sentence} \forall x \forall y: T(x) \land R_\pi(y) \rightarrow \pi(x,y) \label{eq:enum_aux_sentence_AP} \land \\ 
  & \forall x: T(x) \rightarrow (\bigvee_{\pi\in B_\sentence} R_\pi(x)). \label{eq:enum_aux_sentence_A}
\end{align}

Let $\mathcal{T}$ be an interpretation of the target predicate $T$ over $\domain^{\bot}$ such that $T(e)$ is positive if and only if $e$ is the target element $\targete$.
Let $\mathcal{R}_e$ be an interpretation of the relation predicates $R_\pi$ over $e\in \domain^{\bot}$ such that:
\begin{itemize}
  \item If $e$ is the target element, the literal $R_\pi(e)$ is positive if and only if $\pi(e,e) \models \tau(e)$, where $\tau$ is the 1-type of $\targete$ determined in the unary substructure $\usub$.
  \item If $e \in \domain^+$, the literals $R_\pi(e)$ is positive if and only if $\pi$ is the 2-type of $\{e^*, e\}$ in $\substr'$;
  \item Otherwise, all literals $R_\pi(e)$ are negative.
\end{itemize}
Denote by $\mathcal{R} = \bigcup_{e\in \domain^{\bot}} \mathcal{R}_e$ the interpretation of the relation predicates over $\domain^{\bot}$.
We now construct a new substructure $\hat{\substr}$ from $\substr'$:
$$
\hat{\substr} = (\substr' \setminus \substr'_b) \cup \mathcal{T} \cup \mathcal{R}
$$
where $\substr'_b$ is the set of all binary literals in $\substr'$.
Note that $\hat{\substr}$ is a unary substructure of $\tsentence$ over $\domain^{\bot}$.

\begin{example}
  Following \cref{ex:domain_reduction}, the auxiliary sentence $\tsentence_G$ is as follows:
  \begin{align*}
    \tsentence_G = \ &\sentence'_G \ \land \ \forall x \forall y: T(x) \land R_{\pi_E}(y) \rightarrow \pi_E(x,y) \ \land  \\
    &\forall x \forall y: T(x) \land R_{\pi_{\neg E}}(y) \rightarrow \pi_{\neg E}(x,y)  \ \land \\
    &\forall x: T(x)  \rightarrow (R_{\pi_E}(x) \! \lor \! R_{\pi_{\neg E}}(x)). 
  \end{align*}
  After the seventh step in~\cref{fig:example_tree}, the substructure $\hat\substr$ is $\substr'$ $\cup$ $\{T(v_2), \neg T(v_3)\}$ $\cup$ $\{ \neg R_{\pi_{\neg E}}(v_2),$ $\neg R_{\pi_{E}}(v_2),$ $\neg R_{\pi_{\neg E}}(v_3),$ $\neg R_{\pi_{E}}(v_3)\}$.
\end{example}
Then we have the following lemma:
\begin{restatable}{lemma}{unaryencoding}
  \label{le:unary_encoding}
  Any problem $\subsat(\sentence', \domain^{\bot}, \substr')$ resulting from \cref{le:domain_reduction} is equivalent to the problem $\subsat(\tsentence, \domain^{\bot}, \hat{\substr})$.
\end{restatable}
\begin{proof}
  ($\Rightarrow$) Suppose that $\subsat(\sentence', \domain^{\bot}, \substr')$ holds, and let $\fomodel_{\sentence'}$ be any model of $\sentence'$ over $\domain^{\bot}$ such that $\substr' \subseteq \fomodel_{\sentence'}$.
  We construct a structure of $\tsentence$ by adding the unary facts to $\fomodel_{\sentence'}$:
  \begin{equation*}
    \fomodel_{\tsentence} = \fomodel_{\sentence'} \cup \mathcal{T} \cup \mathcal{R}.
  \end{equation*}
  We show that $\fomodel_{\tsentence}\models \tsentence$.
  Clearly, $\fomodel_{\tsentence} \models \sentence'$.
  We only consider the two new parts in \cref{eq:enum_aux_sentence_AP} and \cref{eq:enum_aux_sentence_A}.
  We show that $\fomodel_{\tsentence}$ satisfies \cref{eq:enum_aux_sentence_AP} by checking every pair $\{e, e'\} \in \domain^{\! \bot}$:
  \begin{itemize}
    \item If neither of $e$ and $e'$ is the target element, then $T(x) $ is false, and \cref{eq:enum_aux_sentence_AP} holds;
    \item If both are the target element $\targete$, then the unique 2-type $\pi$ with $\pi(e^*,e^*) = \tau(e^*)$ ensures \cref{eq:enum_aux_sentence_AP} holds;
    \item If exactly one is the target element, without loss of generality, suppose it is $e$:
    \begin{itemize}
      \item If the other element $e'$ is in $\domain^+$, then the 2-type $\pi$ between them corresponds to the unique positive $P_\pi(e')$. 
      Hence, \cref{eq:enum_aux_sentence_AP} holds;
      \item If $e'$ is in $\domain^-$, then all literals $R_\pi(e')$ are negative, and \cref{eq:enum_aux_sentence_AP} still holds.
    \end{itemize}
  \end{itemize}
  Then, for \cref{eq:enum_aux_sentence_A}, we show that it also holds for any $e\in \domain^{\! \bot}$:
  \begin{itemize}
    \item If $e$ is not the target element, then $T(e)$ is negative and~\cref{eq:enum_aux_sentence_A} holds.
    \item If $e$ is the target element, then there is a unique positive literal $R_\pi(e)$ such that $\pi(e,e) = \tau(e)$, and \cref{eq:enum_aux_sentence_A} holds. 
  \end{itemize}
  Thus, $\fomodel_{\tsentence} \models \tsentence$, combining with the fact $\hat{\substr} \subseteq \fomodel_{\tsentence}$, leads to the ``yes'' answer to the problem $\subsat(\tsentence, \domain^{\bot}, \hat{\substr})$.

  ($\Leftarrow$) Conversely, suppose that $\subsat(\tsentence, \domain^{\bot}, \hat{\substr})$ holds, and let $\fomodel_{\tsentence}$ be any model of $\tsentence$ over $\domain^{\bot}$ such that $\hat{\substr} \subseteq \fomodel_{\tsentence}$.
  We construct a structure of $\sentence'$ over $\domain^{\bot}$:
  \begin{equation*}
    \fomodel_{\sentence'} = \fomodel_{\tsentence} \setminus \mathcal{T} \setminus \mathcal{R}.
  \end{equation*}
  Clearly, the difference between $\fomodel_{\tsentence}$ and $\fomodel_{\sentence'}$ is the literals of the target predicate and relation predicates, which do not affect the satisfiability of the sentence $\sentence'$.
  Therefore, $\fomodel_{\sentence'} \models \sentence'$.
  The inclusion $\substr' \subseteq \fomodel_{\sentence'}$ is guaranteed by the definition of $\mathcal{T}$ and $\mathcal{R}$ and the mutual exclusiveness of the 2-types.
\end{proof}

Note that the construction of the problem $\subsat(\tsentence, \domain^{\bot}, \hat{\substr})$ is not unique.
We provide a more efficient alternative in the appendix.

\subsection{From $\subsat$ to $\confsat$}
Following the previous transformations, the substructure decision problem has been reduced to a new problem $\subsat(\tsentence, \domain^{\bot}, \hat{\substr})$, where $\hat{\substr}$ is a \emph{unary} substructure.
Such a problem can be directly reduced to a configuration decision problem.
Let $\tvecn$ be the configuration of $\tsentence$ constructed from the unary substructure $\hat{\substr}$, we have the following lemma.
\begin{lemma}
  \label{le:subsat_to_confsat}
  Any problem $\subsat(\tsentence, \domain^{\bot}, \hat{\substr})$ resulting from \cref{le:unary_encoding} is equivalent to $\confsat(\tsentence, \tvecn)$.
\end{lemma}
\begin{proof}
  The proof is straightforward.
\end{proof}

Combining \cref{le:domain_reduction}, \cref{le:unary_encoding} and \cref{le:subsat_to_confsat}, we have the main conclusion of this section.
\begin{restatable}{proposition}{reducetoconfigurationdecision}
  \label{lem:reduce_to_configuration}
    Any substructure decision problem $\subsat(\sentence, \domain, \substr)$ appearing in \cref{alg:recursive} is equivalent to the configuration decision problem $\confsat(\tsentence, \tvecn)$, where $\tsentence$ and $\tvecn$ are constructed as above.
\end{restatable}

\subsection{Revisiting \cref{alg:recursive}}

By \cref{lem:reduce_to_configuration}, we can check the consistency of a substructure appearing in \cref{alg:recursive} (\cref{line:sat_check}) by solving its corresponding configuration decision problem, which has constant-time complexity, i.e., $\logO(1)$.
However, naively reconstructing the auxiliary configuration $\tvecn$ based on $\substr$ for every consistency check incurs a cost of $\logO(n^2)$, as $\substr$ contains $O(n^2)$ ground literals, which is not efficient.
Observe that $\substr$ is updated incrementally during \cref{alg:recursive}.
We can perform a synchronous update of $\tvecn$ with $\substr$ to avoid this overhead.
The optimized algorithm is presented in \cref{alg:recursive_2}, which maintains a delay complexity of $\logO(n^2)$: 
\begin{restatable}{proposition}{delaycomplexity}
  \label{lem:delay_complexity_2}
  The delay complexity of \cref{alg:recursive_2} is $\logO(n^2)$.
\end{restatable}
\begin{proof}
  The configuration $\tvecn$ will change in four cases:
\begin{enumerate}
  \item When we yield a unary substructure $\usub$, the configuration $\tvecn$ is initialized according to $\usub$ (\cref{line:init_tvecn}).
  \item When we choose a target element $\targete$ from the domain, the configuration $\tvecn$ will be updated, since the interpretation of the target predicate $A$ changes with respect to $\targete$ (\cref{line:update_tvecn_target}).
  \item When we do recursion and try to add a new 2-type to $\substr$ (\cref{line:iterate_2types}), the interpretation of the relation predicates (with respect to $e$) as well as the Tseitin predicates (with respect to $e$ and $\targete$) will be updated, and the configuration $\tvecn$ will be updated accordingly (\cref{line:update_tvecn_2type}). If the updated substructure is not consistent, the configuration $\tvecn$ will be recovered (\cref{line:recover_tvecn_recov}).
  \item Finally, when we determine all 2-types related to $\targete$, the element $\targete$ will be moved from $\domain^{\bot}$ to $\domain^{\top}$, and the configuration $\tvecn$ is updated accordingly (\cref{line:update_tvecn_domain}).
\end{enumerate}
The complexity of the step (2) and (3) is $\logO(1)$, while the complexity of the step (1) and (4) is $\logO(n)$.
However, the step (1) is not in the loop of yielding a binary substructure, and the step (4) only occurs $n-1$ times for enumerating a binary substructure.
Therefore, the total delay complexity of \cref{alg:recursive_2} remains $\logO(n^2)$.
\end{proof}

Combining~\cref{lem:enumeration,lem:delay_complexity_2}, we immediately have that \cref{alg:recursive_2} is a quadratic-delay model enumeration algorithm for \fotwo{} sentences, and thus prove \cref{th:quadratic_delay}.

\begin{algorithm}[tbp]
  \caption{Model Enumeration for \fotwo{}}
  \label{alg:recursive_2}
  \textbf{Input:} An \fotwo{} sentence $\sentence$, a positive integer $n$ \\
  \textbf{Output:} Models of the sentence $\sentence$ over $\domain = \{e_1, \dots, e_n\}$
  \begin{algorithmic}[1]
    \State Set global variables $\sentence$ and $\domain \gets \{e_1, \dots, e_n\}$
    \State \textcolor{blue}{$\tsentence \gets$ The auxiliary sentence for consistency check}
    \For {$\usub$ $\gets$ $\mathsf{EnumUnarySubstructures}(\sentence, n$)}
      \State \textcolor{blue}{$\tvecn \gets $ Initialize the configuration of $\tsentence$ from $\usub$} \label{line:init_tvecn}
      \For {$\fomodel \gets \mathsf{DomainRecursion}(\domain, \usub, \textcolor{blue}{\tvecn})$}
        \State yield $\fomodel$
      \EndFor
    \EndFor

    \State

    \Function{}{}$\mathsf{DomainRecursion}(\domain^{\bot}, \substr, \textcolor{blue}{\tvecn})$
      \If {$|\domain^{\bot}| = 1$}
        \State yield $\substr$
        \State \Return
      \EndIf
      \State $\targete \gets$ Choose a target element from $\domain^{\bot}$
      \State \textcolor{blue}{Update $\tvecn$} \label{line:update_tvecn_target}
      \For {$(\substr, \textcolor{blue}{\tvecn'}) \gets \mathsf{PairRecursion}(\targete, \domain^{\bot}\! \setminus \! \{\targete\}, \substr, \textcolor{blue}{\tvecn})$}
        \State \textcolor{blue}{Update $\tvecn'$} \label{line:update_tvecn_domain}
        \State $\mathsf{DomainRecursion}(\domain^{\bot}\! \setminus \! \{\targete\}, \ \substr, \ \textcolor{blue}{\tvecn'})$
      \EndFor
    \EndFunction

    \State
    
    \Function{}{}$\mathsf{PairRecursion}(\targete, \domain^{\!-}, \substr, \textcolor{blue}{\tvecn}$)
      \If {$\domain^-$ is $\varnothing$}
        \State yield $(\substr, \textcolor{blue}{\tvecn})$ 
        \State \Return
      \EndIf
      \State $e \gets$ Choose an element from $\domain^{\!-}$
      \For {each compatible 2-type $\pi$ of pair $\{e^*, e\}$}  \label{line:iterate_2types}
        \State \textcolor{blue}{$op \gets$ update $\tvecn$} \label{line:update_tvecn_2type}
        \If {\textcolor{blue}{$\confsat(\tsentence, \tvecn)$}}
          \State $\mathsf{PairRecursion}(\targete, \domain^{\!-} \! \setminus \! \{e\}, \substr\cup\{\pi(\targete,e)\}, \textcolor{blue}{\tvecn})$ \label{line:recursive_inner_new}
        \EndIf
        \State \textcolor{blue}{recover $\tvecn$ from $op$} \label{line:recover_tvecn_recov}
      \EndFor
    \EndFunction
  \end{algorithmic}
\end{algorithm}

\section{A Generalization to \fotwo with Equality}
\label{sec:equality}

In this section, we extend the algorithm to the \fotwo{} sentences with the equality predicate.
The equality predicate $=$ is a special binary predicate that is always interpreted as the standard equality relation over the domain.
For convenience, we use $x=y$ to represent the atom $=\!(x,y)$ and $x \neq y$ to represent $\neg(x\!=\!y)$. 
For example, the formula 
\begin{equation}
\label{eq:at_most_one}
\forall x \forall y: \tau_i(x) \land \tau_i(y) \rightarrow x=y
\end{equation}
means that the 1-type $\tau_i$ contains at most one element.
Any \fotwo{} sentence with the equality predicate can also be transformed into an SNF (\cref{eq:snf}) using the same technique by viewing the equality predicate as a normal binary predicate.\footnote{Note that when determining whether a type is compatible, one must also account for the semantics of the equality predicate in addition to the condition defined above.
For example, a 1-type that contains the literal $x \neq x$ cannot be considered compatible. }
W.l.o.g., we assume that the equality predicate only appears in the universally quantified formula $\phi(x,y)$.

Following the same idea as the case without the equality predicate, we first consider the configuration decision problem for the \fotwo{} sentences with equality.
The partial order of configurations is redefined as:
\begin{definition}[Derivation Relation for Equality]
  \label{def:configuration_derivation_with_equality}
  Consider an \fotwo{} sentence $\sentence$ with equality and its configuration $\vecn = (n_1, \dots, n_{|U|})$.
  We say that we can derive a configuration $\vecn' = (n_1', \dots, n_{|U|}')$ from $\vecn$ if
  \begin{align*}
    \left\{
    \begin{aligned}
    &n_i' \ge n_i, \ \text{if } n_i > 1\\
    &n_i' = n_i, \ \text{otherwise}
    \end{aligned}
    \right.
  \end{align*}
  for all $i \in [|U|]$.
\end{definition}
The only difference from the original~\cref{def:configuration_derivation} is that the condition $n_i > 0$ is replaced by $n_i > 1$.
We give an intuitive explanation of why the condition $n_i > 0$ is replaced by $n_i > 1$.
Suppose that $n_i = 1$ and the 1-type $\tau_i$ is restricted to contain at most one element as shown in~\cref{eq:at_most_one}.
When we derive a configuration by increasing the cardinality of $\tau_i$ (i.e., adding a new element to $\tau_i$), then the derived configuration must be unsatisfiable, since there exists no 2-type compatible with the pair $\{\tau_i, \tau_i\}$.
Thus the monotonicity property of the configuration is violated.

However, with the new derivation relation, we can still have \cref{le:configuration_monotonicity} for the \fotwo{} sentences with equality, and the proof is analogous to the case without equality predicates.
Specifically, when extending a model by adding a new element $e$ to a 1-type $\tau_i$ with $n_i > 1$ (as in \cref{sec:foundational_lemmas}), we assign 2-types between $e$ and $\domain \setminus \{e'\}$ by copying the 2-types of an existing element $e'$ in $\tau_i$.
The condition $n_i > 1$ ensures that a 2-type compatible with $\{\tau_i, \tau_i\}$ already exists, thereby guaranteeing the existence of a compatible 2-type between $e$ and $e'$.
In addition, \cref{le:upper_bound} also holds for the \fotwo{} sentences with equality, because the presence of the equality predicate does not affect the original proof.


Consequently, Theorem~\ref{th:1-type_configuration_satisfiability} and Corollary~\ref{cor:configuration_decision} are valid for the \fotwo{} sentences with equality.
Then we can enumerate unary substructures of the sentence $\sentence$ in exactly the same way as before.

Next, recall the process of enumerating binary substructures in Section~\ref{sec:binary_substructure}.
Clearly, the equality predicate does not affect the construction of the binary substructure, where we only care about 2-types between distinct elements.
Therefore, we can directly apply  \cref{alg:recursive_2} for the model enumeration of an \fotwo{} sentence with equality, which proves \cref{th:quadratic_delay_equality}.

\section{Conclusion and Discussion}

In this paper, we establish the quadratic delay complexity of model enumeration for \fotwo{} by introducing a novel algorithm. While the result has practical implications for various applications, such as exhaustive verification and combinatorial structure generation, we emphasize its theoretical significance.

The quadratic delay complexity demonstrates that the model enumeration problem for \fotwo{} is fixed-parameter tractable. This highlights that, at present, model enumeration for \fotwo{} can be far more efficient than its counterparts of model counting and sampling.
We anticipate that this positive result will stimulate more research on the fixed-parameter tractability of model counting and sampling problems.

Additionally, our findings on the configuration decision problem may be of independent interest.
In fact, the configuration decision problem considered in this paper corresponds to the many-sorted spectrum decision problem~\cite{DBLP:journals/jsyml/FischerM04} for many-sorted first-order logic, where the 1-types represent the sorts.
We demonstrate that the well-known MS and ESM properties of \fotwo{} can be generalized to its many-sorted fragment.
This generalization might not be the only case for decidable fragments of first-order logic.
For instance, the spectrum and many-sorted spectrum of \ctwo{} fragment (\fotwo{} with counting quantifiers ($\exists_{=k}, \exists_{\geq k}, \exists_{\leq k}$)) have been also shown to share a similar generalization, both exhibit subtle relationships to the regular graphs.

Extending our enumeration algorithm to more expressive fragments, such as those that have been proven to be tractable for model counting and sampling, while maintaining fixed-parameter tractability is not trivial.
The main barrier lies in the configuration decision problem. 
While the optimized domain recursion technique could still apply, given that all these known tractable fragments are extensions of \fotwo{}, proving the fixed-parameter tractability of model enumeration for these fragments requires solving the configuration decision problem, at least within fixed-parameter polynomial time.

\bibliographystyle{./IEEEtran}
\bibliography{./IEEEabrv,./IEEEexample}

\clearpage

\section*{Appendix}

\subsection{Ommited Proofs in \cref{sec:foundational_lemmas}}

\configurationmonotonicity*
\begin{proof}
  Let $\vecn = (n_1, \dots, n_{|U|})$ be a satisfiable configuration of $\sentence$.
  Since $\vecn$ is satisfiable, there exists a model \fomodel of $\sentence$ such that its configuration is $\vecn$.
  We prove this proposition by constructing a new model $\fomodel'$ of $\sentence$ for any derived configuration $\vecn'$ from $\vecn$ such that the configuration of $\fomodel'$ is $\vecn'$.

  We first reindex the 1-types such that $n_1, \dots, n_d > 0$ and $n_{d+1}, \dots, n_{|U|} = 0$.
  Then this proposition can be restated as follows: if a configuration $\vecn = (n_1, \dots, n_{|U|})$ is satisfiable, then all configurations $(n_1 + l_1, \dots, n_d + l_d, n_{d+1}, \dots, n_{|U|})$, where $l_i \ge 0$ for $i \in [d]$, are also satisfiable.

  We prove it by induction on $l_1, \dots, l_d$.
  The base case where $l_1 = \dots = l_d = 0$ is trivial.
  Suppose that the lemma holds for $k_1, \dots, k_d$, we will show that it also holds for the case where $k_t$ is replaced by $k_t + 1$ for some $t \in [d]$.
  Given a model \fomodel of $\sentence$ under the configuration $\vecn$, we extend the model to $\fomodel'$ by adding a new element $e$ that realizes the 1-type $\tau_t$.
  The 2-types between the new element and other elements are constructed by:
  \begin{enumerate}
    \item picking another element $e'$ realizing $\tau_t$ in \fomodel;
    \item copying the 2-types between $e'$ and other elements in \fomodel to the new element $e$;
    \item and assign the 2-type $\pi$ between $e$ and $e'$ such that $\pi$ is compatible with the 1-type $\tau_t$.
  \end{enumerate}
  It is easy to check the above construction is valid in the sense that we can always find $e'$ by the condition $n_i > 0$, and there is always a compatible 2-type $\pi$ (note that we are not considering the equality predicate here).
  Next, we show that $\fomodel'$ is a model of $\sentence$.
  Let $\domain$ be the domain of \fomodel, and let $\domain'$ be $\domain \cup \{e\}$.
  First, it is easy to check that $\forall x\exists y: \beta_k(x,y)$ for all $k\in[m]$ is satisfied in $\fomodel'$ since 
  \begin{itemize}
    \item the 2-types within $\domain$, which are the same as that in \fomodel, satisfy the formula, and 
    \item the new element $e$ satisfies the formula by the construction step 2.
  \end{itemize}
  For the universal formula $\forall x \forall y: \phi(x,y)$, we ground it to $\domain'$: $\bigwedge_{a, b\in\domain'} \phi(a,b)$.
  We prove its satisfaction by cases:
  \begin{itemize}
    \item If $a, b \in \domain$, then $\phi(a,b)$ holds in \fomodel and thus holds in $\fomodel'$.
    \item If $a = e$ and $b \in \domain\setminus\{e'\}$, then $\phi(e,b)$ is satisfied in $\fomodel'$ by the construction step 2.
    \item If $a = e$ and $b = e'$, then $\phi(e,e')$ is satisfied in $\fomodel'$ by the construction step 3.
  \end{itemize}
  Thus, $\fomodel'$ is a model of $\sentence$, which completes the proof.
\end{proof}

\upperbound*

Recall the notations defined in the main text.
For a given structure $\structure$, we say that an element $e$ is the \emph{(Skolem) witness} of an element $e'$ with respect to $\beta_k$ if the literal $\beta_k(e',e)$ holds in $\structure$.
The element $e$ is also called a \emph{$\beta_k$-witness} of $e'$.
Additionally, if a witness $e$ of $e'$ is in the same 1-type as $e'$, we refer to $e$ as an \emph{internal witness}, or otherwise, an \emph{external witness}.
An element $e$ is considered \emph{satisfied} in a structure $\structure$ if it has at least one witness for each predicate $\beta_k$.
We define the \emph{necessary witness set} of an element $e$ in $\structure$ as a set of witnesses for $e$ that meets the following conditions:
\begin{itemize}
  \item \textbf{Coverage}: for each predicate $\beta_k$, there is at least one $\beta_k$-witness of $e$ in this set.
  \item \textbf{Minimality}: removing any witness from this set results in a set that no longer satisfies the coverage condition.
\end{itemize}
Every element in a model has at least one necessary witness set. 
The size of any necessary witness set for an element is at most $m$, since at most one witness is needed for each $\beta_k$.

Now, we can prove the existence of the universal upper bound $\delta$ for the configuration decision problem.

\begin{proof}
  Without loss of generality, assume that the configuration $\vecn = (n_1, \dots, n_{|U|})$ with $n_1 \geq \delta$ is satisfiable (we can easily guarantee this by re-indexing the 1-types).
  Let \fomodel be a model of the sentence $\sentence$ with the configuration $\vecn$. 
  We will show that $\vecn[1 \mapsto \delta] = (\delta, \dots, n_{|U|})$ is also satisfiable by constructing a new model $\fomodel'$ with $\vecn'$.

  For convenience, we first obtain the necessary witness set for each element in \fomodel, and only consider the witnesses in the necessary witness sets.
  In the construction of $\fomodel'$, we will ensure that each element has at least one witness for each predicate $\beta_k$, following its necessary witness set in \fomodel.
  To distinguish the witnesses in $\fomodel$ and $\fomodel'$, we call the witnesses of an element $e$ in $\fomodel$ as \emph{original witnesses}, and the witnesses in $\fomodel'$ as \emph{new witnesses}.

  Given a structure $\structure$, we view a 1-type $\tau$ as a set of elements realizing $\tau$ in $\structure$.
  We initialize $\fomodel'$ by copying all elements in the 1-types $\tau_2, \dots, \tau_{|U|}$, along with the 2-types between them from \fomodel.
  At this point, $\fomodel'$ may not satisfy $\sentence$, since elements in $\tau_2, \dots, \tau_{|U|}$ may not satisfy $\forall x \exists y: \beta_k(x,y)$ due to the absence of their original witnesses in $\tau_1$.
  In the following, we construct a new set of elements, denoted $\tau_1'$, that realize $\tau_1$ in $\fomodel'$, as well as the 2-types between elements in $\tau_1'$ and other 1-types to ensure that every element is satisfied in $\fomodel'$ (i.e., has a necessary witness set).
  The high-level idea can be described as:
  \begin{enumerate}
    \item Determine 2-types between $\tau_1'$ and other 1-types, ensuring that every element in $\tau_2, \dots, \tau_{|U|}$ is satisfied.
    \item Determine 2-types within $\tau_1'$, ensuring that every element in $\tau_1'$ is satisfied.
  \end{enumerate}
  \paragraph{Part 1}
  Let the set $T$ represent all elements in $\tau_2, \dots, \tau_{|U|}$.
  We create elements in $\tau_1'$ and determine their 2-types with elements in $T$:
  \begin{enumerate}
    \item Choose $m$ arbitrary elements from $\tau_1$ in the original model \fomodel, and denote this set of elements as $S$. 
    \item For an element $e \in T$, find its original witnesses in $\tau_1$ and obtain the corresponding 2-types $\pi_1, \dots, \pi_d$ (note that $d \leq m$).
    For any original witness $a$ of $e$ that is in $S$, maintain the 2-types between $e$ and $a$ in $\fomodel'$.
    For the remaining original witnesses, redirect their 2-types to the remaining elements in $S$ such that for each $\pi_i$, there is an element $a \in S$ with $e$ realizing $\pi_i$.
    As a result, $e$ is satisfied in $\fomodel'$.
    Repeating this process for all elements in $T$ ensures that all elements in $T$ are satisfied.
    \item For each element $e \in S$, find its original witnesses in $T$, and retain their 2-types.
    Here, conflicts may arise since we have already determined 2-types between some elements in $S$ and $T$ during the previous step. 
    For example, if the original $\beta_k$-witness of $e$ is $b\in T$, while $e$ has been assigned as the $\beta_l$-witness of $b$, then we will encounter a conflict if $\beta_k(e,b)$ and $\beta_l(b,a)$ are different (recall that different 2-types are not consistent).
    We denote this conflict as $(e/\beta_k, b/\beta_l)$, which will be resolved by the next step.
    \item For each conflicting pair $(a/\beta_k, b/\beta_l)$ where $a\in S$ and $b\in T$, we resolve the conflict by recovering the 2-type between $a$ and $b$ from \fomodel (i.e., making $b$ the new $\beta_k$-witness of $a$), and adding the original $\beta_l$-witness of $b$ to a new set $S'$.
    Then for each element in $S'$, we directly copy their 2-types with elements in $T$ from the original model \fomodel.
    The combined elements of $S$ and $S'$ form the 1-type $\tau_1'$ in $\fomodel'$.
    \cref{fig:external-witness} provides an illustrative example of this resolution process.
  \end{enumerate}

  \begin{figure}
    \centering
    \includegraphics[width=0.45\textwidth]{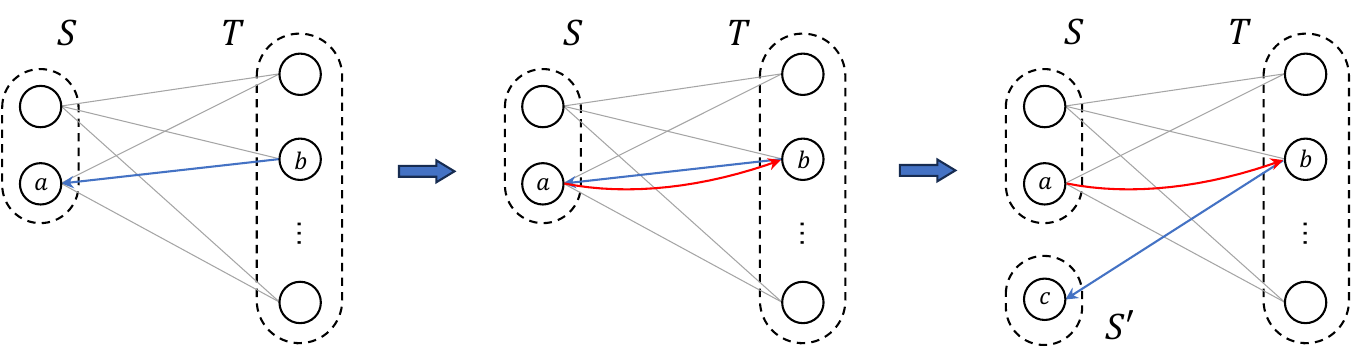}
    \caption{An example of $m=2$: Start by arbitrarily choosing $2$ elements to form $S$ and specifying them as witnesses for elements in $T$. 
    For instance, element $a$ is assigned as the new $\beta_1$-witness of $b$ (blue line) in the step 2.
    Then, when processing the witnesses of $a$, we may find element $b$ as the original $\beta_2$-witness of $a$ (red line) and $\beta_1(b,a)$ and $\beta_2(a,b)$ are different.
    We resolve the conflict by finding the original $\beta_1$-witness of $b$ and add it to $S'$.}
    \label{fig:external-witness}
  \end{figure}
  
  The above construction ensures that for any element $e$ in $T$, and for any original $\beta$-witness of $e$ in $\tau_1$, there is an element in $\tau_1'$ that serves as the new $\beta$-witness of $e$.
  Combined with the initialization of $\fomodel'$, all elements in $T$ are satisfied.
  Additionally, the steps $3-5$ guarantee that the original external witnesses of all elements in $\tau_1'$ are also preserved in $\fomodel'$: for any element $e$ in $\tau_1'$, and for any original external $\beta$-witness of $e$, there is an element in $T$ that serves as the new external $\beta$-witness of $e$.

  Regarding the number of elements in $\tau_1$: In the first step, we add $m$ elements to $S$, in the second step, there are at most $m^2$ conflicts, and in the third step, each conflict will add one element to $S'$. Thus, the number of elements in $\tau_1'$ is at most $m + m^2 = m(m+1)$.

  After the above construction, the number of elements in $\tau_1'$ may not exactly match $\delta$.
  For convenience in the next step, we use the same technique in Lemma \ref{le:configuration_monotonicity} to extend $\tau_1'$ such that $|\tau_1'| = \delta$.

  \paragraph{Part 2}
  We determine the 2-types between elements within $\tau_1'$.
  Let $e_1, \dots, e_{\delta}$ be the elements in $\tau_1'$.
  For each element $e_i$, we find its original internal witnesses, and denote their 2-types by $\{\pi_1, \dots, \pi_k\}$ (note that $k \leq m$).
  Then, we assign the 2-types between $e_i$ and the elements
  $$e_{(i+1)\bmod(2m+1)}, e_{(i+2)\bmod(2m+1)},\dots,e_{(i+k)\bmod(2m+1)}$$
  to be $\{\pi_1, \dots, \pi_k\}$ respectively.
  \cref{fig:internal-witness} provides an example.
  \begin{figure}[tbp]
    \centering
    \includegraphics[width=0.40\textwidth]{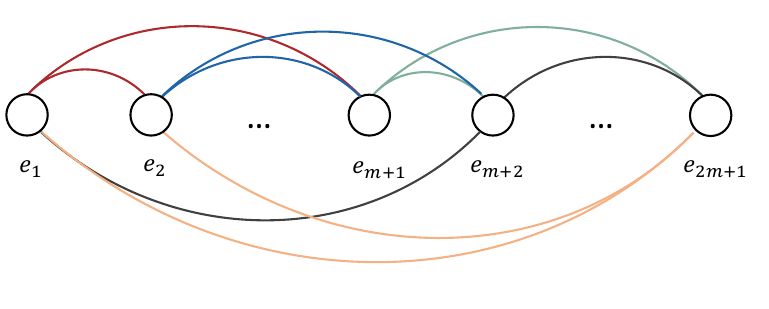}
    \caption{The $m$ new internal witnesses for $e_1$, $e_{m+1}$, $e_{m+2}$ and $e_{2m+1}$ are $\{e_2, e_3, \dots, e_{m+1}\}$, $\{e_{m+2}, \dots, e_{2m+1}\}$, $\{e_{m+3}, \dots, e_{2m+1}, e_1\}$ and $\{e_1, \dots, e_m\}$, respectively.}
    \label{fig:internal-witness}
  \end{figure}
  Then the remaining 2-types within $\tau_1'$ are arbitrarily determined to be compatible with $\tau_1'$.
  This part ensures that all internal witnesses of elements in $\tau_1'$ are preserved in $\fomodel'$, which combined with Part 1, guarantees that all elements in $\tau_1'$ are satisfied in $\fomodel'$.

  Finally, it is straightforward to verify that all elements in $\fomodel'$ satisfy $\forall x\forall y: \phi(x,y)$, since their 2-types are compatible with their respective 1-types.
  Putting all together, we have that $\fomodel'$ is a model of $\sentence$ with the configuration $\vecn'$, i.e., $\vecn'$ is satisfiable.
  The total number of elements in the new 1-type $\tau_1'$ is bounded by $\delta = \max\{m(m+1), 2m+1\}$.
\end{proof}

\subsection{A more efficient construction of the auxiliary sentence}

The new substructure decision problem $\subsat(\tsentence, \domain^{\bot}, \hat{\substr})$ constructed in \cref{le:unary_encoding} is intuitive but not efficient, since we added $|B_\sentence|$ relation predicates $\{P_\pi/1 \mid \pi\in B_\sentence\}$. 
This depends on the number of 2-types in the original sentence, which might be prohibitive in practice.

To this end, we employ a more efficient construction of the auxiliary sentence $\tsentence$ and substructure $\hat\substr$ to reduce the number of relation predicates.
We retain the same unary predicate $T/1$ to identify the target element. 
Instead of introducing $|B_\sentence|$ distinct predicates, we use only one additional unary predicate $R/1$ and one binary predicate $D/2$. 
The atom $R(x)$ indicates whether the 2-type between $x$ and the target element is determined. 
The atom $D(x, y)$ indicates whether either $x$ or $y$ is the target element and the 2-type between them is determined. 
In other words, $D(e_i, e_j)$ holds if and only if $e_i$ or $e_j$ is the target element and the 2-type between $e_i$ and $e_j$ is determined.
The auxiliary sentence is then given by:
\begin{align}
  \tsentence = \ & \sentence' \land \forall x \forall y: T(x) \land R(y) \rightarrow D(x,y) \land D(y,x) \ \land \label{eq:enum_aux_sentence_AP_appendix} \\
  &\forall x \forall y: \neg (T(x) \land R(y)) \land \neg (T(y) \land R(x)) \rightarrow \notag \\
  & \quad\quad\quad\quad\quad\quad\quad\quad\quad\quad \ \ \neg D(x,y) \land \neg D(y,x). \label{eq:enum_aux_sentence_nAP_appendix}
\end{align}
The interpretation of $T/1$ remains as before: $T(e)$ holds if and only if $e = e^*$. 
The predicate $R(e)$ holds if and only if the 2-type between $e$ and the target element is determined, i.e., $e \in {e^*} \cup \domain^{! +}$. 
The interpretation of $D/2$ is uniquely determined by the interpretations of $T$ and $R$ via the constraints in \cref{eq:enum_aux_sentence_AP_appendix,eq:enum_aux_sentence_nAP_appendix}: the literal $D(e, e')$ holds if and only if $T(e)$ and $R(e')$ hold, or $T(e')$ and $R(e)$ hold.

Let $\mathcal{T}, \mathcal{R}$ and $\mathcal{D}$ be the interpretations of the predicates $T, R$ and $D$ over $\domain^{\bot}$, respectively.
The substructure $\hat{\substr}$ is constructed as follows:
$$
\hat{\substr} = (\substr' \setminus \substr'_b) \cup \mathcal{T} \cup \mathcal{R}.
$$
Note that the interpretation of $D$ is not explicitly included in $\hat\substr$, as it is fully determined by $\mathcal{T}$ and $\mathcal{R}$. This ensures that $\hat\substr$ remains a unary substructure, as desired.

Then we proof the following lemma:
\unaryencoding*
\begin{proof}
  ($\Rightarrow$) Suppose that $\subsat(\sentence', \domain^{\bot}, \substr')$ holds, and let $\fomodel_{\sentence'}$ be any model of $\sentence'$ over $\domain^{\bot}$ such that $\substr' \subseteq \fomodel_{\sentence'}$.
  We construct a structure of $\tsentence$ by adding the new facts to $\fomodel_{\sentence'}$:
  \begin{equation*}
    \fomodel_{\tsentence} = \fomodel_{\sentence'} \cup \mathcal{T} \cup \mathcal{R} \cup \mathcal{D}.
  \end{equation*}
  Clearly, $\fomodel_{\tsentence} \models \sentence'$.
  The new parts in \cref{eq:enum_aux_sentence_AP_appendix,eq:enum_aux_sentence_nAP_appendix} only include the literals of the predicates $T$, $R$ and $D$, so $\fomodel_{\tsentence}$ satisfies them by the construction of $\mathcal{T}, \mathcal{R}$ and $\mathcal{D}$.
  Thus, $\fomodel_{\tsentence} \models \tsentence$, combining with the fact $\hat{\substr} \subseteq \fomodel_{\tsentence}$, leads to $\subsat(\tsentence, \domain^{\bot}, \hat{\substr})$ holds.

  ($\Leftarrow$) Conversely, suppose that $\subsat(\tsentence, \domain^{\bot}, \hat{\substr})$ holds, and let $\fomodel_{\tsentence}$ be any model of $\tsentence$ over $\domain^{\bot}$ such that $\hat{\substr} \subseteq \fomodel_{\tsentence}$.
  We construct a structure of $\sentence'$ over $\domain^{\bot}$:
  \begin{equation*}
    \fomodel_{\sentence'} = \fomodel_{\tsentence} \setminus ( \mathcal{T} \cup \mathcal{R} \cup \mathcal{D}).
  \end{equation*}
  Clearly, the difference between $\fomodel_{\tsentence}$ and $\fomodel_{\sentence'}$ is the literals of three new predicates, which do not affect the satisfiability of the sentence $\sentence'$.
  Therefore, $\fomodel_{\sentence'} \models \sentence'$.
  The inclusion $\substr' \subseteq \fomodel_{\sentence'}$ is guaranteed by the definition of $\mathcal{T}$ and $\mathcal{R}$ and the mutual exclusiveness of the 2-types.
\end{proof}

\end{document}